\documentclass[english]{article}
\usepackage[T1]{fontenc}
\usepackage[latin9]{luainputenc}
\usepackage{geometry}
\usepackage{color}
\usepackage{babel}
\usepackage{float}
\usepackage{textcomp}
\usepackage{amsmath}
\usepackage{amsthm}
\usepackage{amssymb}
\usepackage{graphicx}
\usepackage{setspace}
\usepackage[unicode=true,pdfusetitle,
 bookmarks=true,bookmarksnumbered=false,bookmarksopen=false,
 breaklinks=false,pdfborder={0 0 1},backref=false,colorlinks=false]
 {hyperref}

\makeatletter
\theoremstyle{plain}
\newtheorem{lem}{\protect\lemmaname}
\theoremstyle{plain}
\newtheorem{prop}{\protect\propositionname}
\theoremstyle{remark}
\newtheorem{rem}{\protect\remarkname}
\theoremstyle{plain}
\newtheorem{cor}{\protect\corollaryname}

\makeatother

\providecommand{\corollaryname}{Corollary}
\providecommand{\lemmaname}{Lemma}
\providecommand{\propositionname}{Proposition}
\providecommand{\remarkname}{Remark}

\begin{document}
\begin{center}
\textbf{\Large{}Auction Theory Adaptations for Real Life Applications}{\Large\par}
\par\end{center}

\begin{center}
\textbf{Ravi Kashyap}
\par\end{center}

\begin{center}
\textbf{SolBridge International School of Business / City University
of Hong Kong}
\par\end{center}

\begin{center}
\par\end{center}

\begin{center}
\par\end{center}

\begin{center}
\begin{center}
\textbf{\today}
\par\end{center}
\par\end{center}

\begin{center}
\textbf{\textcolor{blue}{\href{https://doi.org/10.1016/j.rie.2018.09.001}{Edited Version: Kashyap, R. (2018). Auction Theory Adaptations for Real Life Applications. Research in Economics, 72(4), 452-481. }}}
\par\end{center}

\begin{center}
\tableofcontents{}
\par\end{center}

\section{\quad Abstract}

We develop extensions to auction theory results that are useful in
real life scenarios.
\begin{enumerate}
\item Since valuations are generally positive we first develop approximations
using the log-normal distribution. This would be useful for many finance
related auction settings since asset prices are usually non-negative.
\item We formulate a positive symmetric discrete distribution, which is
likely to be followed by the total number of auction participants,
and incorporate this into auction theory results.
\item We develop extensions when the valuations of the bidders are interdependent
and incorporate all the results developed into a final combined realistic
setting.
\item Our methods can be a practical tool for bidders and auction sellers
to maximize their profits. The models developed here could be potentially
useful for inventory estimation and for wholesale procurement of financial
instruments and also non-financial commodities.
\end{enumerate}
\textbf{\textcolor{black}{All the propositions are new results and
they refer to existing results which are stated as Lemmas.}}

\section{Auction Strategy}

Auctions are widely used to transact in goods and services, when a
price for the exchange is not readily available before the auction
process is set in motion. Recorded history mentions auctions being
held as early as 500 B.C (Doyle \& Baska 2014; End-note \ref{Auction}).
The collection of works on Auction Theory is vast and deep: (Vickrey
1961) is credited with having initiated the study of auctions as games
of incomplete information. (Fudenberg \& Tirole 1991; Osborne \& Rubinstein
1994) are comprehensive references for the study of games and economic
behavior. 

(Klemperer 1999) is an extensive survey of the enormous literature
regarding auctions; (Milgrom 1989) is an excellent introduction to
the topic and has sufficient technical details to provide a rigorous
yet succinct foundation; other classic surveys are (McAfee \& McMillan
1987a; Wilson 1992). (Wilson 1979) compares the sale prices when the
item is sold to the highest bidder and when the bidders receive fractional
shares of the item at a sale price that equates demand and supply
of shares. Auction situations with multiple related objects being
sold are considered in (Milgrom 1985; Hausch 1986; Armstrong 2000;
Ausubel 2004). 

The core results have been extended to applications in transportation,
telecommunication, network flow assignment, E-commerce, pricing mechanism
for electric power, biological organ transplants (Bertsekas 1988;
Post, Coppinger \& Sheble 1995; McMillan 1995; Tuffin 2002 ; Gregg
\& Walczak 2003; Sheffi 2004; Roth, S�nmez \& �nver 2004), among other
areas, covering both cooperative and competitive aspects of designing
mechanisms to ensure efficient exchange and usage of resources. A
few recent extensions, (Di Corato, Dosi \& Moretto 2017) study how
exit options can affect bidding behavior; (Pica \& Golkar 2017) develop
a framework to evaluate pricing policies of spacecraft trading commodities,
such as data routing, in a federated satellite network. While most
extensions to the basic auction theory results are extremely elegant
and consider complex scenarios; the daily usage of auctions requires
assumptions regarding the number of total bidders and how the valuations
of the different bidders might be distributed. We provide fundamental
extensions, which will be immediately applicable in real life situations.

We consider a few variations in the first price sealed bid auction
mechanism. Once a bidder has a valuation, it becomes important to
consider different auction formats and the specifics of how to tailor
bids, to adapt, to the particular auction setting. A bidding strategy
is sensitive to assumed distributions of both the valuations and the
number of bidders. We build upon existing results from the following
standard and detailed texts on this topic, which act as the foundation:
(Klemperer 2004; Krishna 2009; Menezes \& Monteiro 2005; and Milgrom
2004). Additional references are pointed out in the relevant sections
below along with the extensions that are derived.

The primary distribution we study is the log-normal distribution (Norstad
1999). The log-normal distribution centers around a value and the
chance of observing values further away from this central value become
smaller. Asset prices are generally modeled as log-normal, so financial
applications (Back \& Zender 1993, Nyborg \& Sundaresan 1996 are discussions
about auctions related to treasury securities; Kandel, Sarig \& Wohl
1999, Biais \& Faugeron-Crouzet 2002 analyze the auctions of stock
initial public offerings; Kashyap 2016 is an example of using auctions
for extremely exotic financial products in a complex scenario related
to securities lending) would benefit from this extension, with potential
applications to other goods that follow such a distribution. The absence
of a closed form solution for the log-normal distribution forces us
to develop a rough theoretical approximation and improve upon that
significantly using non linear regressions.

As a simple corollary we consider the uniform distribution, since
the corresponding results can act as a very useful benchmark. The
uniform distribution is well uniform and hence is ideal when the valuations
(or sometimes even the number of bidders) are expected to fall equally
on a finite (or even infinite) number of possibilities. This serves
as one extreme to the sort of distribution we can expect in real life.
The two distribution types we discuss can shed light on the other
types of distributions in which only positive observations are allowed.
The regression based numerical technique that we have developed can
be useful when the valuations are distributed according to other distributions
as well.

We formulate a new positive symmetric discrete distribution, which
is likely to be followed by the total number of auction participants,
and incorporate this into auction theory results. This distribution
can also be a possibility for the valuations themselves, since the
set of prices of assets or valuations can be a finite discrete set.
But given the results for the discrete distribution of auction participants,
developing a bidding strategy based on discrete valuations is trivial
and hence is not explicitly given below. Lastly, the case of interdependent
valuations is to be highly expected in real life (Kashyap 2016 is
an example of such a scenario from the financial services industry);
but practical extensions for this case are near absent. We develop
extensions when the valuations of bidder are interdependent and incorporate
all the results developed into a final combined realistic setting.

\textbf{\textcolor{black}{We provide proofs for the extensions but
we also state standard results with proofs so that it becomes easier
to see how the extensions are developed.}}\textcolor{black}{{} }Such
an approach ensures that the results are instructive and immediately
applicable to both bidders and auction facilitators. \textbf{\textcolor{black}{All
the propositions are new results and they refer to existing results
which are given as Lemmas.}} The results developed here can be an
aid for profit maximization for bidders and auctions sellers during
the wholesale procurement of financial instruments and also non-financial
commodities.

We define all the variables as we introduce them in the text but section
\ref{sec:Dictionary-of-Notation} has a complete dictionary of all
the notation and symbols used in the main results.

\subsection{Symmetric Independent Private Values with Valuations from General
Distribution }

As a benchmark bidding case, it is illustrative to assume that all
bidders know their valuations and only theirs and they believe that
the values of the others are independently distributed according to
the general distribution $F$. $x_{i}$ is the valuation of bidder
$i$. This is a realization of the random variable $X_{i}$ which
bidder $i$ and only bidder $i$ knows for sure. $x_{i}\sim F\left[0,\omega\right]$,
$x_{i}$ is symmetric and independently distributed according to the
distribution $F$ over the interval $\left[0,\omega\right]$. $F,$
is increasing and has full support, which is the non-negative real
line $\left[0,\infty\right]$, hence in this formulation we can have
$\omega=\infty$. $f=F',$ is the continuous density function of $F$.
$M,$ is the total number of bidders. When there is no confusion about
which specific bidder we are referring to, we drop the subscripts
such as in the valuation $x$. $Y_{1}\equiv Y_{1}^{M-1}$ , is the
random variable that denotes the highest value, say for bidder 1,
among the $M-1$ other bidders. $Y_{1},$ is the highest order statistic
of $X_{2},X_{3},...,X_{M}$. $G,$ is the distribution function of
$Y_{1}$. That is, $\forall y,\;G(y)=\left[F(y)\right]^{M-1}$and
$g=G',$ is the continuous density function of $G$ or $Y_{1}$. $m\left(x\right),$
is the expected payment of a bidder with value $x$. $\beta_{i}:\left[0,\omega\right]\rightarrow\Re_{+}$
is an increasing function that gives the strategy for bidder $i$.
We let $\beta_{i}\left(x_{i}\right)=b_{i}$. We must have $\beta_{i}\left(0\right)=0$.
$\beta:\left[0,\omega\right]\rightarrow\Re_{+}$ is the strategy of
all the bidders in a symmetric equilibrium. We let $\beta\left(x\right)=b$.
We also have $b\leq\beta\left(x\right)\;\text{and}\;\beta\left(0\right)=0$.
\begin{lem}
\label{Proposition 3}The symmetric equilibrium bidding strategy for
a bidder, the expected payment of a bidder and the expected revenue
of a seller are given by

Equilibrium Bid Function is,
\begin{eqnarray*}
\beta\left(x\right) & = & \left[x-\int_{0}^{x}\left[\frac{F\left(y\right)}{F\left(x\right)}\right]^{M-1}dy\right]
\end{eqnarray*}
Expected ex ante payment of a particular bidder is,
\begin{eqnarray*}
E\left[m\left(x\right)\right] & = & \int_{0}^{\omega}y\left[1-F\left(y\right)\right]g\left(y\right)dy
\end{eqnarray*}

Expected revenue, $R_{s}$, to the seller is
\[
E\left[R_{s}\right]=ME\left[m\left(x\right)\right]
\]
\end{lem}
\begin{proof}
Appendix \ref{subsec:Proof-of-Lemma-1}.
\end{proof}

\subsection{\label{subsec:Symmetric-Independent-Private-Log-Normal}Symmetric
Independent Private Values with Valuations Distributed Log-normally}
\begin{prop}
\label{The-symmetric-equilibrium-1}The symmetric equilibrium bidding
strategy when the valuations are small, of the order less than one,
and distributed log-normally, can be roughly approximated as
\begin{eqnarray*}
\beta\left(x\right) & = & \left[x-\frac{\int_{0}^{x}\left[\Phi\left(\frac{\ln y-\mu}{\sigma}\right)\right]^{M-1}dy}{\left[\Phi\left(\frac{\ln x-\mu}{\sigma}\right)\right]^{M-1}}\right]\\
 & \approx & \frac{x}{2}
\end{eqnarray*}
Here, $\Phi(u)=\frac{1}{\sqrt{2\pi}}\int_{-\infty}^{u}e^{-t^{2}/2}dt$
, is the standard normal cumulative distribution and $X=e^{W}$where,
$W\sim N\left(\mu,\sigma\right)$. $x_{i}\sim LN\left[0,\omega\right]$
since we are considering the log-normal distribution.
\end{prop}
\begin{proof}
Appendix \ref{subsec:Proof-of-Proposition-3}.
\end{proof}
We see that the log-normal approximation based on the theoretical
approximation (Proposition \ref{The-symmetric-equilibrium-1}), does
not depend on the number of bidders and we provide a possible explanation
for this counter intuitive result below {[}it is interesting to compare
the bidding strategy in the two cases, uniform (Corollary \ref{The-symmetric-equilibrium}
below) and log-normal distribution (Proposition \ref{The-symmetric-equilibrium-1}),
with regards to the number of bidders{]}. This theoretical approximation
is valid only for small values of the parameters of the order of less
than one and the region of validity is limited due to the fact that
the left limit of the bid strategy expression does not exist at $x=0.$

Clearly, numerical integration of the expression for the bid strategy
from (Proposition \ref{The-symmetric-equilibrium-1} or Lemma \ref{Proposition 3})
is a feasible option to obtain the final results. But the drawback
with this approach is that it is not immediately obvious how sensitive
the bid strategy is to the valuation, the parameters of the valuation
distribution and the number of bidders. An in-depth numerical integration
analysis would be able to develop sensitivity measures, but the users
of the final results are not in anyway required to know or even be
familiar with the sensitivities when they use the results. This sensitivity
is highly useful in real-life scenarios when bidders have to estimate
the parameters and hence knowing the sensitivity can possibly help
to counter the extent of uncertainties associated with the estimates.
Also, numerical integration reveals that the theoretical approximation
is quite sensitive to the parameters at different regions of the solution
space.

Hence, to obtain an expression for the bid strategy that depends on
the parameters, we resort to numerically calculating (using numerical
integration) the bid expression for a sample of parameters that are
simulated from appropriate distributions. We draw samples for the
valuation and the valuation distribution parameters from folded normal
distributions. To get the number of bidders, we set an appropriate
ceiling for the values drawn from a folded normal distribution so
that the number of bidders is an integer value greater than or equal
to two. We then run non-linear power regressions, of the form shown
in (eq \ref{eq:Power}), with the numerically calculated bid strategy
values as the dependent variable and the valuation, the valuation
distribution parameters and the number of bidders as independent variables.
\begin{rem}
\label{A-better-approximation}A better approximation can be obtained
using non-linear regression to find the constant, $C$, and the power
coefficients, $a_{1}$,$a_{2}$,$a_{3}$ and $a_{4}$ in the expression
below,
\begin{equation}
\beta\left(x\right)=Cx^{a_{1}}\mu^{a_{2}}\sigma^{a_{3}}M^{a_{4}}\label{eq:Power}
\end{equation}
\end{rem}
The regression coefficients: the constant, $C$, the power coefficients,
$a_{1}$,$a_{2}$,$a_{3}$ and $a_{4}$ in (eq \ref{eq:Power}) are
summarized in figure \ref{fig:Non-Linear-Regression-Co-effecie} (columns:
propConst, valuationP, valMeanP, valStdDevP, biddersMP). The correlation
between the bid strategy values using numerical integration and both
in-sample and out of sample results from using the regression expression
(eq \ref{eq:Power}) are also shown in figure \ref{fig:Non-Linear-Regression-Co-effecie}
(columns: powerCorr, osPCorr), along with the sample size and a variable
that is proportional to the standard deviation of the folded normal
distributions (columns: sampleSize, accuParam). 

\begin{equation}
\beta\left(x\right)=\beta_{0}+\beta_{1}x+\beta_{2}\sigma+\beta_{3}\mu+\beta_{4}\left(M-1\right)+...+\beta_{l}x^{k}+\beta_{l+1}\sigma^{k}+...\label{eq:linear}
\end{equation}

As seen from the above expression (eq \ref{eq:linear}), linear regressions
are not a good choice in this case since some of the regression coefficients
can be negative and this can result in negative values for the bid
strategy. We still report some basic results for the linear regression
(Figure \ref{fig:Non-Linear-Regression-Co-effecie}, column: linearCorr,
gives the correlations between the bid strategy using numerical integration
and results when using a first order regression of the form eq \ref{eq:linear})
since it acts as a benchmark to provide a good comparison for non-linear
regressions, that give better results. Including higher order terms
for the parameters still has the same issue of possibly getting negative
values for the bid strategy and does not increase the accuracy of
the approximation, based on some trials that we conducted.

\begin{figure}[H]
\includegraphics[width=17.5cm]{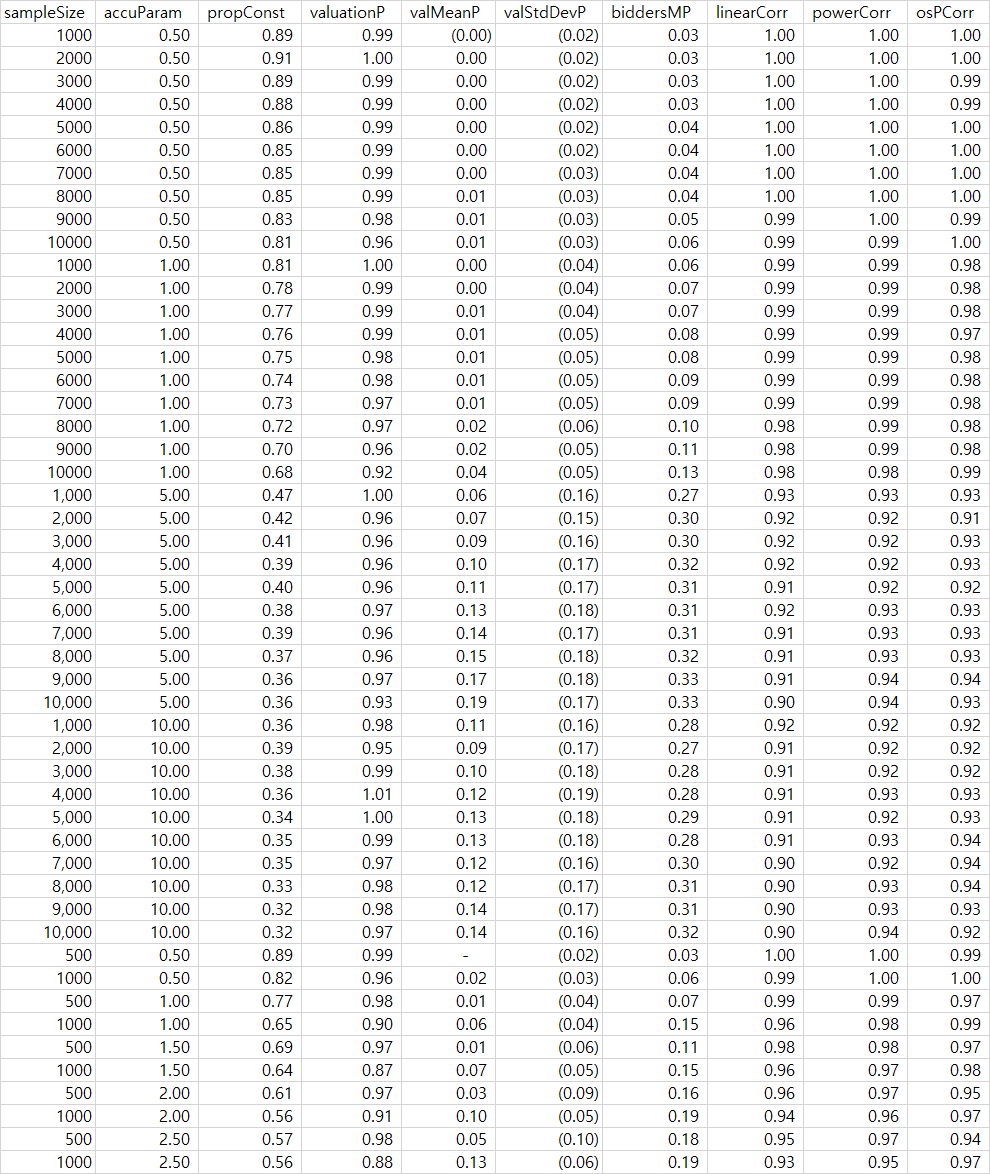}

\caption{\label{fig:Non-Linear-Regression-Co-effecie}Non-Linear Regression
Coefficients and Accuracy}

\end{figure}
 From figure \ref{fig:Non-Linear-Regression-Co-effecie}, we find
that $a_{1}$, which is the coefficient for the valuation $x$, is
the most significant one indicating that the bid strategy is very
sensitive to the valuation and does not depend that much on the other
parameters. This is possibly an explanation for the theoretical approximation
that only involves the valuation in Proposition \ref{The-symmetric-equilibrium-1}.
We see that the power regression produces a more accurate fit than
the linear regression. The power expression is also more accurate
for smaller values of the parameters, say of the order less than one,
and when the number of bidders is more. The lesser the standard deviation
of the valuation distribution, the more accurate the approximation
results. But when the standard deviation of the parameter distributions
is high, the power regression outperform the linear regression. We
do not see a significant difference between the coefficients and the
accuracy of the fit whether we use $M$ or $M-1$ in the regressions.

Figure \ref{fig:Bid-Strategy,-Parameter} gives some raw values of
the bid strategy from the numerical integration (dependent variable)
in the first column and the parameters it depends on, including the
linear regression fit of the first order and the fit from the non-linear
regression; (these independent variables: the valuation, the valuation
distribution mean and standard deviation and the number of bidders
are shown in the subsequent columns, second to fifth columns respectively;
the sixth and seventh columns are the regression results from the
linear and power regressions; the eight column is the ratio of the
numerator and denominator in Proposition \ref{The-symmetric-equilibrium-1},
which is also the equivalent of the valuation minus the bid strategy;
the last three columns give the ratio of the valuation and the bid
strategy using numerical integration, the ratio of the bid strategy
using numerical integration and the power regression result and the
ratio of bid strategy using numerical integration and the linear regression
result). 

Across the entire sample (almost 220,000 observations), the correlation
between the linear regression fit and the bid strategy and the valuation
are 0.34 and 0.35 respectively; the correlation between the power
regression fit and the bid strategy and the valuation are 0.99 and
0.99 respectively. This shows with a certain degree of clarity that
the power regression fit is a better model. A detailed analysis of
the variance explained by both models at different regions of the
parameter space can help in establishing the superiority of the power
model. Such an undertaking can be based on the specifics of the auction
situation and such a more thorough comparison can be useful in deciding
which model to use and related practical considerations.

An immediately obvious fact is that, $\forall x,\beta\left(x\right)\leq x$.
This suggests that fitting different non-linear regression equations
for different buckets of values of $x$ can improve the overall accuracy
of the results. Though we need to be aware of and watch out for the
statistical fact that over-fitting might improve accuracy of the fit
but might not necessarily improve the predictive power. Hence, we
need to perform out of sample tests for any models we fit before we
deem the models to be providing an improvement. Lastly, it should
be clear that the regression based technique is very powerful and
can be applicable in many situations. Fitting a regression equation
of the form in (Eq. \ref{eq:Power}) can be useful even when the valuations
follow other complex distributions, with suitable modifications to
factor in the corresponding distribution parameters as independent
variables. We could even introduce other independent variables to
capture the auction environment such as the industry, the extent of
competition within an industry, any constraints being faced by the
auction sellers and so on.

\begin{figure}[H]
\includegraphics[width=17.5cm]{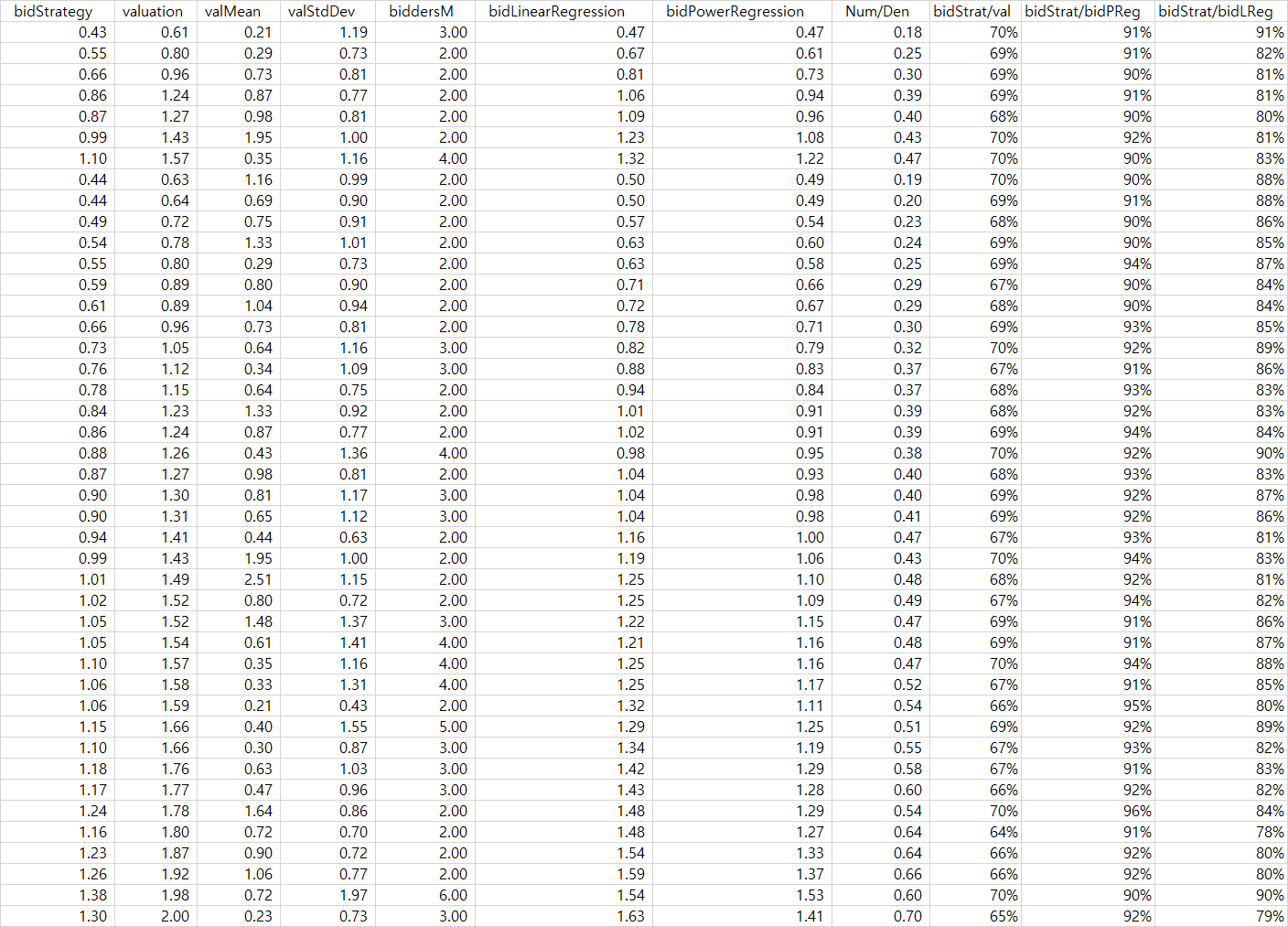}

\caption{\label{fig:Bid-Strategy,-Parameter}Bid Strategy, Parameter Samples
and Regression Results}

\end{figure}

\subsection{Symmetric Independent Private Values with Valuations Distributed
Uniformly}
\begin{cor}
\label{The-symmetric-equilibrium}The symmetric equilibrium bidding
strategy when the valuations are distributed uniformly is given by
\begin{eqnarray*}
\beta\left(x\right) & = & \left(\frac{M-1}{M}\right)x
\end{eqnarray*}
Here, $x_{i}\sim U\left[0,\omega\right]$ since we are considering
the uniform distribution.
\end{cor}
\begin{proof}
Appendix \ref{subsec:Proof-of-Proposition-2}.
\end{proof}
When the number of bidders are large the above expression for the
uniform distribution does not depend on the number of bidders, that
is $\underset{M\rightarrow\infty}{\lim}\left(\frac{M-1}{M}\right)x=x$.
Comparing the bidding strategy with respect to the distribution of
valuations in the two cases, the uniform distribution when the number
of bidders are large and the log-normal distribution theoretical approximation
(Proposition \ref{The-symmetric-equilibrium-1}) we see that: 1) both
do not depend significantly on the number of bidders and 2) the bid
is larger with a uniform distribution.

\subsection{Symmetric Independent Private Value with Reserve Prices}
\begin{lem}
\label{Proposition 6}The symmetric equilibrium bidding strategy when
the valuation is greater than the reserve price, $\left(r>0\right)$,
of the seller, $x\geq r$, is

For a general distribution, 
\begin{eqnarray*}
\beta\left(x\right) & = & r\frac{G\left(r\right)}{G\left(x\right)}+\frac{1}{G\left(x\right)}\int_{r}^{x}yg\left(y\right)dy\\
\text{Alternately }\beta\left(x\right) & = & x-\int_{r}^{x}\frac{G\left(y\right)}{G\left(x\right)}dy
\end{eqnarray*}
\end{lem}
\begin{proof}
Appendix \ref{subsec:Proof-of-Lemma-4}.
\end{proof}

\subsubsection{Uniform Distributions }
\begin{cor}
\label{The-symmetric-equilibrium-2}The symmetric equilibrium bidding
strategy when the valuation is greater than the reserve price of the
seller, $x\geq r$, and valuations are from an uniform distribution,
\[
\beta\left(x\right)=\frac{r^{M}}{x^{M-1}}\left(\frac{M+1}{M}\right)+x\left(\frac{M-1}{M}\right)
\]
\end{cor}
\begin{proof}
Appendix \ref{subsec:Proof-of-Proposition-5}.
\end{proof}

\subsubsection{Log-normal Distributions }
\begin{cor}
\label{The-symmetric-equilibrium-3}The symmetric equilibrium bidding
strategy when the valuation is greater than the reserve price of the
seller, $x\geq r$, and valuations are from a log-normal distribution,
\[
\beta\left(x\right)=x\left[\frac{h'\left(r\right)\left(x-r\right)}{h\left(x\right)}+\frac{r}{x}\frac{h\left(r\right)}{h\left(x\right)}\right]
\]
\[
\text{Here, }h'\left(r\right)=\left(M-1\right)\left[\int_{-\infty}^{\left(\frac{\ln r-\mu}{\sigma}\right)}e^{-t^{2}/2}dt\right]^{M-2}\left\{ \frac{e^{-\left(\frac{\ln r-\mu}{\sigma}\right)^{2}/2}}{r\sigma}\right\} 
\]
\end{cor}
\begin{proof}
Appendix \ref{subsec:Proof-of-Proposition-6}.
\end{proof}
We can obtain decent approximations using non-linear power regressions,
of the form shown in (eq \ref{eq:Power}), with the numerically calculated
bid strategy values as the dependent variable and the valuation, the
valuation distribution parameters, the reserve price and the number
of bidders as independent variables.

\subsubsection{Optimal Reserve Price for Seller and Other Considerations}
\begin{lem}
\label{The-optimal-reserve}The optimal reserve price for the seller,
$r^{*}$ must satisfy the following expression, 
\[
x_{s}=r^{*}-\frac{\left[1-F\left(r^{*}\right)\right]}{f\left(r^{*}\right)}
\]

Here, seller has a valuation, $x_{s}\in\left[0,\omega\right)$ t
\end{lem}
\begin{proof}
Appendix \ref{subsec:Proof-of-Lemma-7}.
\end{proof}

\subsection{Variable Number of Bidders with Symmetric Valuations and Beliefs
about Number of Bidders}

$\mathcal{M}=\left\{ 1,2,\ldots,M\right\} $ is the potential set
of bidders when there is uncertainty about how many interested bidders
there are. $\mathcal{A}\subseteq\mathcal{M}$ is the set of actual
bidders. All potential bidders draw their valuations independently
distributed according to the general distribution $F$. Also, $p_{l}$
is the probability that any participating bidder, $i\in\mathcal{A},$
is facing $l$ other bidders or the probability that he assigns to
the event that he is facing $l$ other bidders. This implies that
there is a total of $l+1$ bidders, $l\in\left\{ 1,2,\ldots,M-1\right\} $.
$G^{l}\left(x\right)=\left[F\left(x\right)\right]^{l}$ is the probability
of the event that the highest of $l$ values drawn from the symmetric
distribution $F$ is less than $x$, his valuation and the bidder
wins in this case. $\beta^{l}\left(x\right)$ is the equilibrium bidding
strategy when there are a total of exactly $l+1$ bidders, known with
certainty. The overall probability that the bidder will win when he
bids $\beta^{M}\left(x\right)$ is 
\[
G\left(x\right)=\sum_{l=0}^{M-1}p_{l}G^{l}\left(x\right)
\]

Hence the equilibrium bid for an actual bidder when he is unsure about
the number of rivals he faces is a weighted average of the equilibrium
bids in an auction when the number of bidders is known to all. (McAfee
\& McMillan 1987b) is one of the most well known and early generalizations
to allow the number of bidders to be stochastic.
\begin{lem}
\label{The-equilibrium-strategy}The equilibrium strategy when there
is uncertainty about the number of bidders is given by
\begin{eqnarray*}
\beta^{M}\left(x\right) & = & \sum_{l=0}^{M-1}\frac{p_{l}G^{l}\left(x\right)}{G\left(x\right)}\beta^{l}\left(x\right)
\end{eqnarray*}
\end{lem}
\begin{proof}
Appendix \ref{subsec:Proof-of-Lemma-8}.
\end{proof}
Any financial market participant or intermediary, will expect most
of the other major players to be bidding at an auction as well. Invariably,
there will be some drop outs, depending on their recent bidding activity
and some smaller players will show up depending on the specifics of
the particular auction. It is a reasonable assumption that all of
the bidders hold similar beliefs about the distribution of the number
of players. We construct a symmetric discrete distribution of the
sort shown in the diagram below (Figure \ref{fig:Variable-Bidders-Symmetric}).
We show that such a distribution satisfies all the properties of a
probability distribution function as part of the proof for Proposition
\ref{The-formula-for}. It is to be noted that this symmetric discrete
distribution comes under the family of triangular distributions (End-note
\ref{Triangular-Distribution}). We can easily come up with variations
that can provide discrete asymmetric probabilities. For simplicity,
we use the uniform distribution for the valuations and set $\omega=1$.
The below result follows from a bidding strategy that incorporates
the use of the discrete symmetric distribution.
\begin{prop}
\label{The-formula-for}The bidding strategy and the formula for the
probability of facing any particular total number of bidders under
a symmetric discrete distribution would be given by,
\[
\beta\left(x\right)=\sum_{l=0}^{M-1}\left(\frac{p_{l}x^{l}}{\sum_{k=0}^{M-1}p_{k}x^{k}}\right)\left(\frac{l}{l+1}\right)x
\]
\begin{eqnarray*}
p_{l} & = & \begin{cases}
l\Delta_{p} & ,\;\text{if}\;\;l\leq\frac{\left(M-1\right)}{2}\\
\left(M-l\right)\Delta_{p} & ,\;\text{if}\;\;l>\frac{\left(M-1\right)}{2}
\end{cases}
\end{eqnarray*}
\[
\Delta_{p}=\left\{ \left\lfloor \frac{M^{2}}{4}\right\rfloor \right\} ^{-1}
\]
We note that $\Delta_{p}$ can also be written as,
\[
\Delta_{p}=\frac{1}{\left\{ \left\lfloor \frac{\left(M-1\right)}{2}\right\rfloor \left\{ \left\lfloor \frac{\left(M-1\right)}{2}\right\rfloor +1\right\} +\left[\left\{ \left(\frac{\left(M-1\right)}{2}\bmod1\right)+\frac{\left(M-1\right)}{2}\right\} \left\{ 2\left(\frac{\left(M-1\right)}{2}\bmod1\right)\right\} \right]\right\} }
\]
\begin{align*}
\left\lfloor \frac{\left(M-1\right)}{2}\right\rfloor \text{ is the integer floor function, that is, it rounds any number down to the nearest integer. }\\
A\bmod B\text{ is the modulo operator, that is, it gives the remainer when }A\text{ is divided by }B.\\
\text{When }A\text{ is a fraction less than one and }B\text{ is one, the result is the fraction itself.}
\end{align*}
\end{prop}
\begin{proof}
Appendix \ref{subsec:Proof-of-Proposition-Discrete-Symmetric}.
\end{proof}
\begin{figure}[H]
\includegraphics[width=8cm]{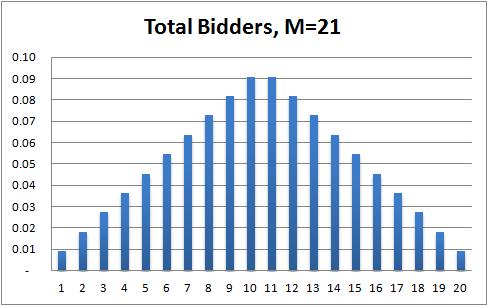}\includegraphics[width=8cm]{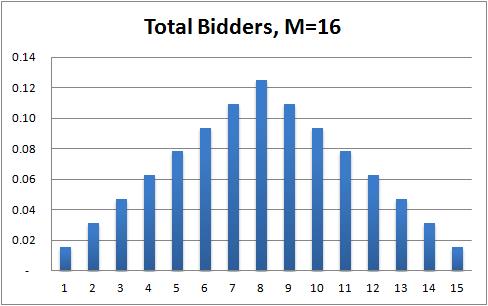}

\caption{\label{fig:Variable-Bidders-Symmetric}Variable Bidders Symmetric
Discrete Probability Distributions}
\end{figure}

\subsection{Asymmetric Valuations}

$f_{i},\,F_{i}$ , are the continuous density function and distribution
of bidder $i$ in this case where the valuations are asymmetric. $\phi_{i}\equiv\beta_{i}^{-1}$
is the inverse of the bidding strategy $\beta_{i}$. This means, $x_{i}=\beta_{i}^{-1}\left(b_{i}\right)=\phi_{i}\left(b_{i}\right)$.
The following result captures the scenario when we have an asymmetric
equilibrium.
\begin{lem}
\label{The-system-of}The system of differential equations for an
asymmetric equilibrium is given by
\begin{eqnarray*}
\sum_{j\neq i}^{j\in\left\{ 1,...,M\right\} }\left\{ \frac{f_{j}\left(\phi_{j}\left(b\right)\right)\phi'_{j}\left(b\right)}{F_{j}\left(\phi_{j}\left(b\right)\right)}\right\}  & = & \frac{1}{\left[\phi_{i}\left(b\right)-b\right]}
\end{eqnarray*}
\end{lem}
\begin{proof}
Appendix \ref{subsec:Proof-of-Lemma-9}.
\end{proof}
This system of differential equations can be solved to get the bid
functions for each player. Closed form solutions are known for the
case of uniform distributions with different supports. A simplification
is possible by assuming that say, some bidders have one distribution
and some others have another distribution. This is a reasonable assumption
since financial firms can be categorized into big global and small
local companies. Such an assumption can hold even for telecommunication
companies or other sectors where local players compete with global
players to obtain rights for service to a region or to obtain use
of raw materials from a region.
\begin{prop}
\label{If,--firms}If, $K+1$ firms (including the one for which we
derive the payoff condition) have the distribution $F_{1}$, strategy
$\beta_{1}$ and inverse function $\phi_{1}$. The other $M-K-1$
firms have the distribution $F_{2}$, strategy $\beta_{2}$ and inverse
function $\phi_{2}$. The system of differential equations is given
by,

\begin{eqnarray*}
\left\{ K\frac{f_{1}\left(\phi_{1}\left(b\right)\right)\phi'_{1}\left(b\right)}{\left[F_{1}\left(\phi_{1}\left(b\right)\right)\right]}\right\} +\left\{ \left(M-1-K\right)\frac{f_{2}\left(\phi_{2}\left(b\right)\right)\phi'_{2}\left(b\right)}{\left[F_{2}\left(\phi_{2}\left(b\right)\right)\right]}\right\} =\frac{1}{\left[\phi_{i}\left(b\right)-b\right]}
\end{eqnarray*}
\end{prop}
\begin{proof}
Appendix \ref{subsec:Proof-of-Proposition-10}.
\end{proof}
As a special case, if there are only two bidders, $M=2,K=1$ the above
reduces to a system of two differential equations, 
\begin{eqnarray*}
\phi'_{1}\left(b\right) & = & \frac{\left[F_{1}\left(\phi_{1}\left(b\right)\right)\right]}{f_{1}\left(\phi_{1}\left(b\right)\right)\left[\phi_{2}\left(b\right)-b\right]}\\
\phi'_{2}\left(b\right) & = & \frac{\left[F_{2}\left(\phi_{2}\left(b\right)\right)\right]}{f_{2}\left(\phi_{2}\left(b\right)\right)\left[\phi_{1}\left(b\right)-b\right]}
\end{eqnarray*}

\subsection{Symmetric Interdependent Valuations}

It is worth noting that a pure common value model of the sort, $V=\upsilon\left(X_{1},X_{2},...,X_{M}\right)$
is not relevant for many entities (firms or individuals) since forces
that drive them to participate in auctions might vary. $X_{i}\in\left[0,\omega_{i}\right]$
is bidder $i's$ signal when the valuations are interdependent. $V_{i}=\upsilon_{i}\left(X_{1},X_{2},...,X_{M}\right)$
is the item value to bidder $i$ and $\upsilon_{i}\left(0,0,...,0\right)=0$.
$\upsilon_{i}\left(x_{1},x_{2},...,x_{M}\right)\equiv E\left[V_{i}\mid X_{1}=x_{1},X_{2}=x_{2},...,X_{M}=x_{M},\right]$
is a more general setting, where knowing the signals of all bidders
still does not reveal the full value with certainty. 

In the case of financial firms (Kashyap 2016) the motivation to participate
in auctions will depend on their trading activity, which will not
be entirely similar, but we can expect many common drivers of trading
activity. What this reasoning tells us is that it is reasonable to
expect that there is some correlation between the signals of each
bidder. This results in a symmetric interdependent auction strategy.
From the perspective of a particular bidder, the signals of the other
bidders can be interchanged without affecting the value. This is captured
using the function $u\left(X_{i},X_{-i}\right)$ which is the same
for all bidders and is symmetric in the last $M-1$ components. We
assume that all signals $X_{i}$ are from the same distribution $\left[0,\omega\right]$
and that the valuations can be written as 
\[
\upsilon_{i}\left(X_{1},X_{2},...,X_{M}\right)=u\left(X_{i},X_{-i}\right)
\]
We also assume that the joint density function of the signals $f$
defined on $\left[0,\omega\right]^{M}$ is symmetric and the signals
are affiliated. Affiliation here refers to the below properties.
\begin{itemize}
\item The random variables $X_{1},X_{2},...,X_{M}$ distributed on some
product of intervals $\mathcal{X}\subset\Re^{M}$ according to the
joint density function $f$. The variables $\mathbf{X=}\left(X_{1},X_{2},...,X_{M}\right)$
are affiliated if $\forall\mathbf{x',x''}\in\mathcal{X}$, $f\left(\mathbf{x'\lor x''}\right)f\left(\mathbf{x'\land x''}\right)\geq f\left(\mathbf{x'}\right)f\left(\mathbf{x''}\right)$.
Here $\mathbf{x'\lor x''}$ and $\mathbf{x'\lor x''}$ denote the
component wise maximum and minimum of $\mathbf{x'}$ and $\mathbf{x''}$.
\item The random variables $Y_{1},Y_{2},...,Y_{M-1}$ denote the largest,
second largest, ... , smallest from among $X_{2},X_{3},...,X_{M}$.
If $X_{1},X_{2},...,X_{M}$ are affiliated, then $X_{1},Y_{1},Y_{2},...,Y_{M-1}$
are also affiliated. 
\item Let $G\left(.\mid x\right)$ denote the distribution of $Y_{1}$ conditional
on $X_{1}=x$ and let $g\left(.\mid x\right)$ be the associated conditional
density function. Then if $Y_{1}$ and $X_{1}$ are affiliated and
if $x'>x$ then $G\left(.\mid x'\right)$ dominates $G\left(.\mid x\right)$
in terms of the reverse hazard rate, $\frac{g\left(t\right)}{G\left(t\right)}$.
That is $\forall y$,
\[
\frac{g\left(y\mid x'\right)}{G\left(y\mid x'\right)}\geq\frac{g\left(y\mid x\right)}{G\left(y\mid x\right)}
\]
\item If $\gamma$ is any increasing function, then $x'>x$ implies that
\[
E\left[\gamma\left(Y_{1}\right)\mid X=x'\right]\geq E\left[\gamma\left(Y_{1}\right)\mid X=x\right]
\]
\end{itemize}
We define the below function as the expectation of the value to bidder
$1$ when the signal he receives is $x$ and the highest signal among
the other bidders, $Y_{1}=y$. Because of symmetry this function is
the same for all bidders and we assume it is strictly increasing in
$x$. We also have $u\left(\mathbf{0}\right)=\upsilon(0,0)=0$. 
\[
\upsilon\left(x,y\right)=E\left[V_{1}\mid X=x,Y_{1}=y\right]
\]

\begin{lem}
\label{A-symmetric-equilibrium}A symmetric equilibrium strategy governed
by the set of conditions above is given by
\[
\beta\left(x\right)=\int_{0}^{x}\upsilon\left(y,y\right)dL\left(y\mid x\right)
\]

Here, we define $L\left(y\mid x\right)$ as a function with support
$\left[0,\omega\right]$,
\[
L\left(y\mid x\right)=\exp\left[-\int_{y}^{x}\frac{g\left(t\mid t\right)}{G\left(t\mid t\right)}dt\right]
\]
\end{lem}
\begin{proof}
Appendix \ref{subsec:Proof-of-Lemma-11}.
\end{proof}
\begin{prop}
\label{The-bidder's-equilibrium}The bidder's equilibrium strategy
under a scenario when the valuation is the weighted average of his
valuation and the highest of the other valuations is given by the
expression below. That is, we let $\upsilon\left(x,y\right)=\alpha x+\xi y\;\mid\alpha,\xi\in\left[0,1\right]$.
This also implies, $\upsilon\left(x,y\right)=u\left(x,y\right)=u\left(x_{i},x_{-i}\right)=\alpha x_{i}+\xi\max\left(x_{-i}\right)$,
giving us symmetry across the signals of other bidders. An alternative
formulation could simply be $\upsilon\left(x,y\right)=\frac{1}{M}\left(\overset{M}{\underset{i=1}{\sum}}x_{i}\right)$.
The affiliation structure follows the Irwin-Hall distribution (End-note
\ref{Irwin-Hall}) with bidder's valuation being the sum of a signal
coming from a uniform distribution with $\omega=1$ and a common component
from the same uniform distribution. 
\end{prop}
\begin{eqnarray*}
\beta\left(x\right) & = & \left[\frac{2\left(\alpha+\xi\right)\left(M-1\right)}{\left(2M-1\right)x^{2M-2}}\right]\\
 & + & \left(\alpha+\xi\right)\left[x-\frac{1}{\left(2x-1-\frac{x^{2}}{2}\right)^{M-1}}\left\{ \frac{1}{2^{M-1}}+\int_{1}^{x}\left(2y-1-\frac{y^{2}}{2}\right)^{M-1}dy\right\} \right]
\end{eqnarray*}

\begin{proof}
The proof in Appendix \ref{subsec:Proof-of-Proposition-12} includes
a method to solve the last integral.
\end{proof}
Despite the simplifications, regarding the distribution assumptions
in Proposition \ref{The-bidder's-equilibrium}, being satisfactory
approximations in many instances, they have been made to ensure that
the results are analytically tractable. More complex distributions
can be accommodated using the regression technique we have developed
in section \ref{subsec:Symmetric-Independent-Private-Log-Normal}.
A regression equation of the form in (Eq. \ref{eq:Power}) can be
used with suitable modifications to introduce other distribution parameters
as independent variables.

\subsection{Combined Realistic Setting}
\begin{prop}
\label{The-bidding-strategy}The bidding strategy in a realistic setting
with symmetric interdependent, uniformly distributed valuations, with
reserve prices and variable number of bidders is given by 
\end{prop}
\[
\beta\left(x\right)=re^{-\int_{x^{*}}^{x}\frac{g\left(t\mid t\right)}{G\left(t\mid t\right)}dt}+\int_{x^{*}}^{x}v(y,y)\frac{g\left(y\mid y\right)}{G\left(y\mid y\right)}e^{-\int_{y}^{x}\frac{g\left(t\mid t\right)}{G\left(t\mid t\right)}dt}dy
\]

Here, $x^{*}\left(r\right)$ is found by solving for $x$ in the below
condition 
\[
\left[\int_{0}^{1}\xi y\left[\frac{y}{x}\right]^{2\left(M-2\right)}\left(\frac{2y}{x^{2}}\right)dy+\int_{1}^{x}\xi y\left[\frac{\left(2y-1-\frac{y^{2}}{2}\right)}{\left(2x-1-\frac{x^{2}}{2}\right)}\right]^{M-2}\left\{ \frac{2-y}{\left(2x-1-\frac{x^{2}}{2}\right)}\right\} dy\right]=\frac{r-\alpha x}{\left(M-1\right)}
\]

\begin{proof}
The proof in Appendix \ref{subsec:Proof-of-Proposition-13} includes
a method to solve the above type of equations. 
\end{proof}
It is trivial to extend the above to the case where the total number
of bidders is uncertain by using the equilibrium bidding strategy
$\beta^{l}\left(x\right)$ and the associated probability $p_{l}$
when there are exactly $l+1$ bidders, known with certainty, 
\[
\beta^{M}\left(x\right)=\sum_{l=0}^{M-1}\frac{p_{l}G^{l}\left(x\right)}{G\left(x\right)}\beta^{l}\left(x\right)
\]

\section{Improvements to the Model}

We have assumed that the valuation each bidder holds is derived independently
before he decides upon his bidding strategy. Instead of the bidding
strategies we have considered, we can come up with a parametric model
that will take the valuations as the inputs and the bid as output.
From a financial industry securities lending perspective (Kashyap
2016), the parameters can depend on the characteristics of the portfolio
being auctioned such as, the size of the portfolio, the number of
securities, the number of markets, the extent of overlap with the
internal inventory, and where available, the percentile rankings of
the historical bids for previous auctions, which auction sellers do
reveal sometimes. It should be acknowledged, however, that this would
be useful if and only if the parametric model is a valid approximation
of the general model and the actual auction environment, which may
be non-parametric. Therefore, a parametrization may limit the general
applicability of any results obtained, though it can be useful for
specific applications where the auction setting and the parameters
are known to a satisfactory extent. The similarity of the parametric
approach to using a regression equation when the valuations follow
a particular distribution are to be noted. This suggests that the
regression technique can be useful in a variety of situations, where
closed form solutions are either absent or not readily obtained.

(Cobb, Rumi \& Salmer�n 2012; and Nie \& Chen 2007) derive approximate
distributions for the sum of log-normal distributions which highlight
that we can estimate the log-normal parameters from the time series
of the valuations and hence get the mean and variance of the valuations.
(Norstad 1999) is a basic but excellent reference for the log-normal
distribution. A longer historical time series will help get better
estimates for the volatility of the valuation. This can be useful
to decide the aggressiveness of the bid. Another key extension can
be to introduce jumps in the log-normal processes. This is seen in
stock prices to a certain extent and to a greater extent in the inventory
processes.

\section{Conclusion}

We have looked at various auction strategy extensions that would be
relevant to the application of auctions to many real life situations.
\textbf{\textcolor{black}{All the propositions are new results and
they refer to existing results which are given as Lemmas.}} We have
derived the closed form solutions for bidding strategies, relevant
to an auction, where such a formulation exists and in situations where
approximations would be required, we have provided those. A simple
result using the log-normal distribution should be of immense practical
use, along with a new positive symmetric distribution that can be
used to handle situations where both the valuations and number of
bidders are distributed accordingly. The interdependent valuations
scenario and the combined realistic setting provide ready to use results
for bidders and auction designers alike. These extensions are immediately
useful for the auctions of financial securities and should be applicable
for numerous other products. 

\section{Acknowledgements and End-Notes}
\begin{enumerate}
\item Dr. Yong Wang, Dr. Isabel Yan, Dr. Vikas Kakkar, Dr. Fred Kwan, Dr.
William Case, Dr. Srikant Marakani, Dr. Qiang Zhang, Dr. Costel Andonie,
Dr. Jeff Hong, Dr. Guangwu Liu, Dr. Andrew Chan, Dr. Humphrey Tung
and Dr. Xu Han at the City University of Hong Kong provided advice
and more importantly encouragement to explore and where possible apply
cross disciplinary techniques. The faculty members of SolBridge International
School of Business provided patient guidance and valuable suggestions
on how to further this paper.
\item Numerous seminar participants provided helpful suggestions. The views
and opinions expressed in this article, along with any mistakes, are
mine alone and do not necessarily reflect the official policy or position
of either of my affiliations or any other agency.
\item \label{Auction}\href{https://en.wikipedia.org/wiki/Auction}{Auction, Wikipedia Link}:
An auction is a process of buying and selling goods or services by
offering them up for bid, taking bids, and then selling the item to
the highest bidder.
\item \label{Triangular-Distribution}\href{https://en.wikipedia.org/wiki/Triangular_distribution}{Triangular distribution, Wikipedia Link}:
In probability theory and statistics, the triangular distribution
is a continuous probability distribution with lower limit $a$, upper
limit $b$ and mode $c$, where $a<b$ and $a\leq c\leq b$ (also
see: Evans, Hastings \& Peacock 2000).
\item \label{Irwin-Hall}\href{https://en.wikipedia.org/wiki/Irwin\%E2\%80\%93Hall_distribution}{Irwin\textendash Hall distribution, Wikipedia Link}:
In probability and statistics, the Irwin--Hall distribution, named
after Joseph Oscar Irwin and Philip Hall, is a probability distribution
for a random variable defined as the sum of a number of independent
random variables, each having a uniform distribution. For this reason
it is also known as the uniform sum distribution (also see: Hall 1927;
Irwin 1927).
\end{enumerate}

\section{References}
\begin{itemize}
\item Armstrong, M. (2000). Optimal multi-object auctions. The Review of
Economic Studies, 67(3), 455-481.
\item Ausubel, L. M. (2004). An efficient ascending-bid auction for multiple
objects. The American Economic Review, 94(5), 1452-1475.
\item Back, K., \& Zender, J. F. (1993). Auctions of divisible goods: On
the rationale for the treasury experiment. The Review of Financial
Studies, 6(4), 733-764.
\item Bertsekas, D. P. (1988). The auction algorithm: A distributed relaxation
method for the assignment problem. Annals of operations research,
14(1), 105-123.
\item Biais, B., \& Faugeron-Crouzet, A. M. (2002). IPO auctions: English,
Dutch,\dots{} French, and internet. Journal of Financial Intermediation,
11(1), 9-36.
\item Chiani, M., Dardari, D., \& Simon, M. K. (2003). New exponential bounds
and approximations for the computation of error probability in fading
channels. Wireless Communications, IEEE Transactions on, 2(4), 840-845.
\item Cobb, B. R., Rumi, R., \& Salmer�n, A. (2012). Approximating the Distribution
of a Sum of Log-normal Random Variables. Statistics and Computing,
16(3), 293-308.
\item Di Corato, L., Dosi, C., \& Moretto, M. (2017). Multidimensional auctions
for long-term procurement contracts with early-exit options: the case
of conservation contracts. European Journal of Operational Research.
\item Doyle, R., \& Baska, S. (2014). History of Auctions: From ancient
Rome to today's high-tech auctions. Auctioneer, SUA.
\item Dyer, D., Kagel, J. H., \& Levin, D. (1989). Resolving uncertainty
about the number of bidders in independent private-value auctions:
an experimental analysis. The RAND Journal of Economics, 268-279.
\item Evans, M., Hastings, N., \& Peacock, B. (2000). Triangular distribution.
Statistical distributions, 3, 187-188.
\item Fudenberg, D., \& Tirole, J. (1991). Game theory (Vol. 1). The MIT
press.
\item Gregg, D. G., \& Walczak, S. (2003). E-commerce auction agents and
online-auction dynamics. Electronic Markets, 13(3), 242-250.
\item Hall, P. (1927). The distribution of means for samples of size N drawn
from a population in which the variate takes values between 0 and
1, all such values being equally probable. Biometrika, Vol. 19, No.
3/4., pp. 240--245.
\item Harstad, R. M., Kagel, J. H., \& Levin, D. (1990). Equilibrium bid
functions for auctions with an uncertain number of bidders. Economics
Letters, 33(1), 35-40.
\item Hausch, D. B. (1986). Multi-object auctions: Sequential vs. simultaneous
sales. Management Science, 32(12), 1599-1610.
\item Huang, H. N., Marcantognini, S., \& Young, N. (2006). Chain rules
for higher derivatives. The Mathematical Intelligencer, 28(2), 61-69.
\item Irwin, J. O. (1927). On the frequency distribution of the means of
samples from a population having any law of frequency with finite
moments, with special reference to Pearson's Type II. Biometrika,
Vol. 19, No. 3/4., pp. 225--239.
\item Johnson, W. P. (2002). The curious history of Fa� di Bruno's formula.
American Mathematical Monthly, 217-234.
\item Kandel, S., Sarig, O., \& Wohl, A. (1999). The demand for stocks:
An analysis of IPO auctions. The Review of Financial Studies, 12(2),
227-247.
\item Kashyap, R. (2016). Securities Lending Strategies, Exclusive Valuations
and Auction Bids. Social Science Research Network. Working Paper.
\item Klemperer, P. (1999). Auction theory: A guide to the literature. Journal
of economic surveys, 13(3), 227-286.
\item Klemperer, P. (2004). Auctions: theory and practice.
\item Krishna, V. (2009). Auction theory. Academic press.
\item Laffont, J. J., Ossard, H., \& Vuong, Q. (1995). Econometrics of first-price
auctions. Econometrica: Journal of the Econometric Society, 953-980.
\item Lebrun, B. (1999). First price auctions in the asymmetric N bidder
case. International Economic Review, 40(1), 125-142.
\item Levin, D., \& Ozdenoren, E. (2004). Auctions with uncertain numbers
of bidders. Journal of Economic Theory, 118(2), 229-251.
\item McAfee, R. P., \& McMillan, J. (1987a). Auctions and bidding. Journal
of economic literature, 25(2), 699-738.
\item McAfee, R. P., \& McMillan, J. (1987b). Auctions with a stochastic
number of bidders. Journal of Economic Theory, 43(1), 1-19.
\item McMillan, J. (1995). Why auction the spectrum?. Telecommunications
policy, 19(3), 191-199.
\item Menezes, F. M., \& Monteiro, P. K. (2005). An introduction to auction
theory (pp. 1-178). Oxford University Press.
\item Milgrom, P. R., \& Weber, R. J. (1982). A theory of auctions and competitive
bidding. Econometrica: Journal of the Econometric Society, 1089-1122.
\item Milgrom, P. R. (1985). The economics of competitive bidding: a selective
survey. Social goals and social organization: Essays in memory of
Elisha Pazner, 261-292.
\item Milgrom, P. (1989). Auctions and bidding: A primer. The Journal of
Economic Perspectives, 3(3), 3-22.
\item Milgrom, P. R. (2004). Putting auction theory to work. Cambridge University
Press.
\begin{doublespace}
\item Miranda, M. J., \& Fackler, P. L. (2002). Applied Computational Economics
and Finance.
\end{doublespace}
\item Nie, H., \& Chen, S. (2007). Lognormal sum approximation with type
IV Pearson distribution. IEEE Communications Letters, 11(10), 790-792.
\item Norstad, John. \textquotedbl The normal and lognormal distributions.\textquotedbl{}
(1999).
\item Nyborg, K. G., \& Sundaresan, S. (1996). Discriminatory versus uniform
Treasury auctions: Evidence from when-issued transactions. Journal
of Financial Economics, 42(1), 63-104.
\item Ortega-Reichert, A. (1967). Models for competitive bidding under uncertainty.
Stanford University.
\item Osborne, M. J., \& Rubinstein, A. (1994). A course in game theory.
MIT press.
\item Pica, U., \& Golkar, A. (2017). Sealed-Bid Reverse Auction Pricing
Mechanisms for Federated Satellite Systems. Systems Engineering.
\item Post, D. L., Coppinger, S. S., \& Sheble, G. B. (1995). Application
of auctions as a pricing mechanism for the interchange of electric
power. IEEE Transactions on Power Systems, 10(3), 1580-1584.
\item Roth, A. E., S�nmez, T., \& �nver, M. U. (2004). Kidney exchange.
The Quarterly Journal of Economics, 119(2), 457-488.
\item Sheffi, Y. (2004). Combinatorial auctions in the procurement of transportation
services. Interfaces, 34(4), 245-252.
\item Tuffin, B. (2002). Revisited progressive second price auction for
charging telecommunication networks. Telecommunication Systems, 20(3),
255-263.
\item Vickrey, W. (1961). Counterspeculation, auctions, and competitive
sealed tenders. The Journal of finance, 16(1), 8-37.
\item Wilson, R. (1979). Auctions of shares. The Quarterly Journal of Economics,
675-689.
\item Wilson, R. (1992). Strategic analysis of auctions. Handbook of game
theory with economic applications, 1, 227-279.
\end{itemize}

\section{\label{sec:Dictionary-of-Notation}Dictionary of Notation and Terminology
for the Auction Strategy}
\begin{itemize}
\item $x_{i}$, the valuation of bidder $i$. This is a realization of the
random variable $X_{i}$ which bidder $i$ and only bidder $i$ knows
for sure.
\item $x_{i}\sim F\left[0,\omega\right]$, $x_{i}$ is symmetric and independently
distributed according to the distribution $F$ over the interval $\left[0,\omega\right]$. 
\item $F,$ is increasing and has full support, which is the non-negative
real line $\left[0,\infty\right]$.
\item $f=F',$ is the continuous density function of $F$.
\item $x_{i}\sim U\left[0,\omega\right]$ when we consider the uniform distribution.
\item $x_{i}\sim LN\left[0,\omega\right]$ when we consider the log-normal
distribution.
\item $M,$ is the total number of bidders.
\item $f_{i},\,F_{i}$ , are the continuous density function and distribution
of bidder $i$ in the asymmetric case.
\item $r\geq0$, is the reserve price set by the auction seller.
\item $\beta_{i}:\left[0,\omega\right]\rightarrow\Re_{+}$ is an increasing
function that gives the strategy for bidder $i$. We let $\beta_{i}\left(x_{i}\right)=b_{i}$.
We must have $\beta_{i}\left(0\right)=0$.
\item $\phi_{i}\equiv\beta_{i}^{-1}$ is the inverse of the bidding strategy
$\beta_{i}$. This means, $x_{i}=\beta_{i}^{-1}\left(b_{i}\right)=\phi_{i}\left(b_{i}\right)$.
\item $x_{i}\sim F_{i}\left[0,\omega_{i}\right]$. Here, $x_{i}$ is asymmetric
and is independently distributed according to the distribution $F_{i}$
over the interval $\left[0,\omega_{i}\right]$. 
\item $\beta:\left[0,\omega\right]\rightarrow\Re_{+}$ is the strategy of
all the bidders in a symmetric equilibrium. We let $\beta\left(x\right)=b,x$
is the valuation of any bidder. We also have $b\leq\beta\left(x\right)\;\text{and}\;\beta\left(0\right)=0$.
\item $Y_{1}\equiv Y_{1}^{M-1}$ , the random variable that denotes the
highest value, say for bidder 1, among the $M-1$ other bidders.
\item $Y_{1},$ is the highest order statistic of $X_{2},X_{3},...,X_{M}$.
\item $G,$ is the distribution function of $Y_{1}$. $\forall y,\;G(y)=\left[F(y)\right]^{M-1}$.
\item $g=G',$ is the continuous density function of $G$ or $Y_{1}$.
\item $\Pi_{i},$ is the payoff of bidder $i$. $\Pi_{i}=\begin{cases}
x_{i}-b_{i}\quad if\;b_{i}>max_{j\neq i}b_{j}\\
0\qquad\quad if\;b_{i}<max_{j\neq i}b_{j}
\end{cases}$
\item $\Pi_{s},\;x_{s}$ is the payoff and valuation of the auction seller. 
\item $m\left(x\right),$ is the expected payment of a bidder with value
$x$.
\item $R_{s}$ is the expected revenue to the seller.
\item $\mathcal{M}=\left\{ 1,2,\ldots,M\right\} $ is the potential set
of bidders when there is uncertainty about how many interested bidders
there are. $\mathcal{A}\subseteq\mathcal{M}$ is the set of actual
bidders.
\item $p_{l}$ is probability that any participating bidder assigns to the
event that he is facing $l$ other bidders or that there is a total
of $l+1$ bidders, $l\in\left\{ 1,2,\ldots,M-1\right\} $. 
\item $X_{i}\in\left[0,\omega_{i}\right]$ is bidder $i's$ signal when
the valuations are interdependent. 
\item $V_{i}=\upsilon_{i}\left(X_{1},X_{2},...,X_{M}\right)$ is the item
value to bidder $i$. $\upsilon_{i}\left(0,0,...,0\right)=0$
\item $\upsilon_{i}\left(x_{1},x_{2},...,x_{M}\right)\equiv E\left[V_{i}\mid X_{1}=x_{1},X_{2}=x_{2},...,X_{M}=x_{M},\right]$
is a more general setting, where knowing the signals of all bidders
still does not reveal the full value with certainty.
\item $\because$ used in some of the proofs means ``because''.
\end{itemize}

\section{\label{sec:Appendix-of-Proofs}Appendix of Proofs}

\textbf{\textcolor{black}{All the propositions are new results and
they refer to existing results which are given as Lemmas.}}

\subsection{\label{subsec:Proof-of-Lemma-1}Proof of Lemma \ref{Proposition 3}.}

The proof from (Krishna 2009) is repeated below for completeness.
\begin{proof}
Symmetry among the bidders implies, $b_{i}=\beta_{i}\left(x_{i}\right)=\beta\left(x_{i}\right)$.
Say bidder 1 has valuation $x$ and submits bid $b$. Bidder 1 wins
whenever he submits the highest bid, that is whenever, 
\begin{eqnarray*}
b & > & max_{i\neq1}\beta\left(X_{i}\right)\\
 & > & \beta\left(max_{i\neq1}X_{i}\right)\;\left[\because\beta\;is\;increasing.\right]\\
 & > & \beta\left(Y_{1}\right)\\
\Rightarrow Y_{1} & < & \beta^{-1}\left(b\right)
\end{eqnarray*}
Expected Payoff of bidder 1 then becomes,
\[
\Pi_{1}=G\left(\beta^{-1}\left(b\right)\right)\left\{ x-b\right\} 
\]
First Order Conditions to maximize the payoff then give,
\begin{eqnarray*}
\frac{\partial\left\{ G\left(\beta^{-1}\left(b\right)\right)\left\{ x-b\right\} \right\} }{\partial b} & = & 0\\
\Rightarrow g\left(\beta^{-1}\left(b\right)\right)\frac{\partial\left\{ \beta^{-1}\left(b\right)\right\} }{\partial b}\left(x-b\right)-G\left(\beta^{-1}\left(b\right)\right) & = & 0
\end{eqnarray*}
Define $z=\beta^{-1}\left(b\right)\Rightarrow b=\beta\left(z\right)$.
Differentiating with respect to $b$ gives,
\begin{eqnarray*}
1 & = & \frac{\partial\beta\left(z\right)}{\partial z}\frac{\partial z}{\partial b}\\
\Rightarrow\frac{\partial z}{\partial b} & \equiv & \frac{\partial\left\{ \beta^{-1}\left(b\right)\right\} }{\partial b}=\frac{1}{\beta'\left(z\right)}=\frac{1}{\beta'\left(\beta^{-1}\left(b\right)\right)}
\end{eqnarray*}
Using this we have,
\begin{eqnarray*}
\frac{g\left(\beta^{-1}\left(b\right)\right)}{\beta'\left(\beta^{-1}\left(b\right)\right)}\left(x-b\right)-G\left(\beta^{-1}\left(b\right)\right) & = & 0
\end{eqnarray*}
At a symmetric equilibrium we have $b=\beta\left(x\right)$, 
\begin{eqnarray*}
\Rightarrow\frac{g\left(x\right)}{\beta'\left(x\right)}\left(x-b\right)-G\left(x\right) & = & 0\\
bg\left(x\right)+\beta'\left(x\right)G\left(x\right) & = & xg\left(x\right)\\
\beta\left(x\right)g\left(x\right)+\beta'\left(x\right)G\left(x\right) & = & xg\left(x\right)\\
\Rightarrow\frac{d\left(G\left(x\right)\beta\left(x\right)\right)}{dx} & = & xg\left(x\right)
\end{eqnarray*}
Integrating this from $0$ to $x$ and using $\beta\left(0\right)=0$,
\begin{eqnarray*}
\beta\left(x\right) & = & \frac{1}{G\left(x\right)}\int_{0}^{x}yg\left(y\right)dy\\
 & = & E\left[Y_{1}\mid Y_{1}<x\right]
\end{eqnarray*}
Using the formula for integration by parts, 
\begin{eqnarray*}
\int_{0}^{x}yg\left(y\right)dy & = & \int_{0}^{x}yG'\left(y\right)dy\\
 & = & \left|yG\left(y\right)\right|_{0}^{x}-\int_{0}^{x}G\left(y\right)dy
\end{eqnarray*}
\begin{eqnarray*}
\beta\left(x\right) & = & \frac{1}{G\left(x\right)}\left[xG\left(x\right)-\int_{0}^{x}G\left(y\right)dy\right]\\
 & = & \left[x-\int_{0}^{x}\frac{G\left(y\right)}{G\left(x\right)}dy\right]\\
\beta\left(x\right) & = & \left[x-\int_{0}^{x}\left[\frac{F\left(y\right)}{F\left(x\right)}\right]^{M-1}dy\right]
\end{eqnarray*}
It is easily shown that this is the best response for any bidder,
as follows. Say, Bidder 1 bids $b'=\beta\left(z\right)$ when his
valuation is $x$. His expected payoff is then given by 
\begin{eqnarray*}
\Pi\left(b',x\right) & = & G\left(z\right)\left[x-\beta\left(z\right)\right]\equiv\left(Prob\;of\;Win\right)*\left(Payoff\right)\\
 & = & G\left(z\right)x-G\left(z\right)\left[z-\int_{0}^{z}\frac{G\left(y\right)}{G\left(z\right)}dy\right]\;,\;using\;\beta\left(z\right)\;from\;above\\
 & = & G\left(z\right)x-G\left(z\right)z+\int_{0}^{z}G\left(y\right)dy\\
\Pi\left(\beta\left(z\right),x\right)\equiv\Pi\left(b',x\right) & = & G\left(z\right)\left(x-z\right)+\int_{0}^{z}G\left(y\right)dy
\end{eqnarray*}
Considering,
\begin{eqnarray*}
\Pi\left(\beta\left(x\right),x\right)-\Pi\left(\beta\left(z\right),x\right) & = & \left[G\left(x\right)\left(x-x\right)+\int_{0}^{x}G\left(y\right)dy\right]\\
 &  & -\left[G\left(z\right)\left(x-z\right)+\int_{0}^{z}G\left(y\right)dy\right]\\
 & = & G\left(z\right)\left(z-x\right)+\int_{0}^{x}G\left(y\right)dy-\left[\int_{0}^{x}G\left(y\right)dy+\int_{x}^{z}G\left(y\right)dy\right]\\
 & = & G\left(z\right)\left(z-x\right)-\int_{x}^{z}G\left(y\right)dy\\
 & \geq & 0,\;when\;z\geq x\;or\;z\leq x
\end{eqnarray*}
Expected payment by a bidder with value $x$ is,
\begin{eqnarray*}
m\left(x\right) & = & Prob\left[Win\right]*\left(Amount\;Bid\right)\\
 & = & G\left(x\right)*\beta\left(x\right)\;\left[\because G\left(x\right)\equiv Prob\left[Y_{1}<x\right],Y_{1}\equiv Highest\;of\;\left(M-1\right)\;other\;bids\right]\\
 & = & \int_{0}^{x}yg\left(y\right)dy
\end{eqnarray*}
Expected ex ante payment of a particular bidder is,
\begin{eqnarray*}
E\left[m\left(x\right)\right] & = & \int_{0}^{\omega}m\left(x\right)f\left(x\right)dx\\
 & = & \int_{0}^{\omega}\left(\int_{0}^{x}yg\left(y\right)dy\right)f\left(x\right)dx\\
 & = & \int_{0}^{\omega}\left(\int_{y}^{\omega}f\left(x\right)dx\right)yg\left(y\right)dy\;\quad Changing\;order\;of\;integration\\
 & = & \int_{0}^{\omega}\left[F\left(\omega\right)-F\left(y\right)\right]yg\left(y\right)dy\\
 & = & \int_{0}^{\omega}y\left[1-F\left(y\right)\right]g\left(y\right)dy\;\left(\because F\left(\omega\right)=1\right)
\end{eqnarray*}
 Expected revenue to the seller is sum of payments of all the bidders.
\[
E\left[R_{s}\right]=ME\left[m\left(x\right)\right]
\]
\end{proof}

\subsection{\label{subsec:Proof-of-Proposition-3}Proof of Proposition \ref{The-symmetric-equilibrium-1}.}
\begin{proof}
Using the bid function from Lemma (\ref{Proposition 3}) with the
log-normal distribution functions, $F\left(x\right)=\Phi\left(\frac{lnx-\mu}{\sigma}\right),G\left(x\right)=\left(\Phi\left(\frac{lnx-\mu}{\sigma}\right)\right)^{M-1}$.
Here, $\Phi(u)=\frac{1}{\sqrt{2\pi}}\int_{-\infty}^{u}e^{-t^{2}/2}dt$
, is the standard normal cumulative distribution and $X=e^{W}$where,
$W\sim N\left(\mu,\sigma\right)$
\begin{eqnarray*}
\beta\left(x\right) & = & \left[x-\int_{0}^{x}\left[\frac{F\left(y\right)}{F\left(x\right)}\right]^{M-1}dy\right]\\
 & = & \left[x-\int_{0}^{x}\left[\frac{\Phi\left(\frac{lny-\mu}{\sigma}\right)}{\Phi\left(\frac{lnx-\mu}{\sigma}\right)}\right]^{M-1}dy\right]
\end{eqnarray*}
No closed form solution is available. There are certain approximations,
which can be used, (See Laffont, Ossard \& Vuong 1995). We provide
a simplification using the Taylor series expansion as shown below.
This is valid only for non zero values of $x$ (The Taylor series
for this function is undefined at $x=0$, but we consider the right
limit to evaluate this at zero), which holds in our case since a zero
valuation will mean a zero bid. 
\begin{eqnarray*}
\beta\left(x\right) & = & \left[x-\frac{\int_{0}^{x}\left[\Phi\left(\frac{lny-\mu}{\sigma}\right)\right]^{M-1}dy}{\left[\Phi\left(\frac{lnx-\mu}{\sigma}\right)\right]^{M-1}}\right]\\
 & = & \left[x-\frac{\left\{ \int_{0}^{x}\left[\frac{1}{\sqrt{2\pi}}\int_{-\infty}^{\left(\frac{lny-\mu}{\sigma}\right)}e^{-t^{2}/2}dt\right]^{M-1}dy\right\} }{\left[\frac{1}{\sqrt{2\pi}}\int_{-\infty}^{\left(\frac{lnx-\mu}{\sigma}\right)}e^{-t^{2}/2}dt\right]^{M-1}}\right]
\end{eqnarray*}
Let,
\begin{eqnarray*}
h\left(y\right) & = & \left[\int_{-\infty}^{\left(\frac{lny-\mu}{\sigma}\right)}e^{-t^{2}/2}dt\right]^{M-1}\\
j\left(y\right) & = & \int h\left(y\right)dy
\end{eqnarray*}
We then have,
\begin{eqnarray*}
\beta\left(x\right) & = & \left[x-\frac{\left\{ \int_{0}^{x}h\left(y\right)dy\right\} }{h\left(x\right)}\right]\\
 & = & \left[x-\frac{\left|j\left(y\right)\right|_{0}^{x}}{h\left(x\right)}\right]=\left[x-\left\{ \frac{j\left(x\right)-j\left(0\right)}{h\left(x\right)}\right\} \right]\\
 & \approx & \left[x-\frac{j'\left(0\right)x}{h\left(x\right)}\right]\;\left\{ \because j\left(x\right)-j\left(0\right)\simeq j'\left(0\right)x\;,\quad Maclaurin\;Series\right\} \\
 & = & \left[x-\frac{h\left(0\right)x}{h\left(x\right)}\right]=x\left[1-\frac{h\left(0\right)}{h\left(x\right)}\right]
\end{eqnarray*}
\begin{eqnarray*}
\Rightarrow\beta\left(x\right) & = & x\;\left[\;\because h\left(0\right)=\left[\int_{-\infty}^{\left(\frac{ln0-\mu}{\sigma}\right)}e^{-t^{2}/2}dt\right]^{M-1}=\left[\int_{-\infty}^{-\infty}e^{-t^{2}/2}dt\right]^{M-1}=0\right]
\end{eqnarray*}
We could include additional terms, for greater precision, using the
subsequent terms of the Maclaurin series, as follows,
\begin{eqnarray*}
\beta\left(x\right) & \approx & x\left[1-\frac{h\left(0\right)}{h\left(x\right)}-\frac{x}{2}\frac{h'\left(0\right)}{h\left(x\right)}\right]\;\left\{ \because j\left(x\right)-j\left(0\right)\simeq j'\left(0\right)x+\frac{j''\left(0\right)x^{2}}{2!}\right\} \\
\beta\left(x\right) & = & x\left[1-\frac{h\left(0\right)}{h\left(x\right)}-\frac{1}{2}\left\{ 1-\frac{h\left(0\right)}{h\left(x\right)}\right\} \right]\;\left\{ \because\frac{h\left(x\right)-h\left(0\right)}{h\left(x\right)}\simeq\frac{h'\left(0\right)x}{h\left(x\right)}\right\} \\
\Rightarrow\beta\left(x\right) & = & \frac{x}{2}\;\left[\;\because h\left(0\right)=0\right]
\end{eqnarray*}
\begin{eqnarray*}
\beta\left(x\right) & \approx & x\left[1-\frac{h\left(0\right)}{h\left(x\right)}-\frac{x}{2}\frac{h'\left(0\right)}{h\left(x\right)}-\frac{x^{2}}{6}\frac{h''\left(0\right)}{h\left(x\right)}\right]\\
 &  & \;\left\{ \because j\left(x\right)-j\left(0\right)\simeq j'\left(0\right)x+\frac{j''\left(0\right)x^{2}}{2!}+\frac{j'''\left(0\right)x^{3}}{3!}\right\} 
\end{eqnarray*}
\begin{eqnarray*}
\beta\left(x\right) & \approx & x\left[1-\frac{h\left(0\right)}{h\left(x\right)}-\frac{x}{2}\frac{h'\left(0\right)}{h\left(x\right)}-\frac{1}{3}\left\{ 1-\frac{h\left(0\right)}{h\left(x\right)}-x\frac{h'\left(0\right)}{h\left(x\right)}\right\} \right]\\
 &  & \;\left\{ \because\frac{h\left(x\right)-h\left(0\right)-h'\left(0\right)x}{h\left(x\right)}\simeq\frac{1}{2}\frac{h''\left(0\right)x^{2}}{h\left(x\right)}\right\} 
\end{eqnarray*}
\begin{eqnarray*}
\beta\left(x\right) & \approx & x\left[1-\frac{h\left(0\right)}{h\left(x\right)}-\frac{x}{2}\frac{h'\left(0\right)}{h\left(x\right)}-\frac{1}{3}+\frac{1}{3}\frac{h\left(0\right)}{h\left(x\right)}+\frac{x}{3}\frac{h'\left(0\right)}{h\left(x\right)}\right]\\
 & = & x\left[\frac{2}{3}-\frac{2}{3}\frac{h\left(0\right)}{h\left(x\right)}-\frac{1}{6}\left\{ 1-\frac{h\left(0\right)}{h\left(x\right)}\right\} \right]\\
\Rightarrow\beta\left(x\right) & = & \frac{x}{2}\;\left[\;\because h\left(0\right)=0\right]
\end{eqnarray*}

Checking the Lagrange remainders $R^{M}\left(y\right)$ for a degree
$M$ approximation where $0<\xi_{M}<y$, 
\begin{eqnarray*}
R^{M}\left(y\right) & = & \frac{j^{M+1}\left(\xi_{M}\right)y^{M+1}}{\left(M+1\right)!}=\left|\frac{\partial^{M+1}j\left(y\right)}{\partial y^{M+1}}\right|_{y=\xi_{M}}\left[\frac{y^{M+1}}{\left(M+1\right)!}\right]\\
 & = & \left|\frac{\partial^{M+1}\left\{ \int\left[\int_{-\infty}^{\left(\frac{lny-\mu}{\sigma}\right)}e^{-t^{2}/2}dt\right]^{M-1}\right\} }{\partial y^{M+1}}\right|_{y=\xi_{M}}\left[\frac{y^{M+1}}{\left(M+1\right)!}\right]
\end{eqnarray*}
Let,
\begin{eqnarray*}
l\left(y\right) & = & \left(\frac{lny-\mu}{\sigma}\right)\Rightarrow l'\left(y\right)=\frac{1}{\sigma y}\Rightarrow l''\left(y\right)=-\frac{1}{\sigma y^{2}}
\end{eqnarray*}
\begin{eqnarray*}
R^{M}\left(y\right) & = & \left[\frac{y^{M+1}}{\left(M+1\right)!}\right]\left|\frac{\partial^{M}\left\{ \left[\int_{-\infty}^{l\left(y\right)}e^{-t^{2}/2}dt\right]^{M-1}\right\} }{\partial y^{M}}\right|_{y=\xi_{M}}
\end{eqnarray*}

Using Fa'adi Bruno's Formula (Huang, Marcantognini \& Young 2006;
Johnson 2002), 
\[
\frac{\partial^{M}\left\{ p\left(q\left(y\right)\right)\right\} }{\partial y^{M}}=\sum\frac{M!}{k_{1}!\;k_{2}!\;...\;k_{M}!}p^{\left(k\right)}\left(q\left(y\right)\right)\left(\frac{q'\left(y\right)}{1!}\right)^{k_{1}}\left(\frac{q''\left(y\right)}{2!}\right)^{k_{2}}...\left(\frac{q^{\left(M\right)}\left(y\right)}{M!}\right)^{k_{M}}
\]
where the sum is over all non-negative integer solutions of the Diophantine
equation $k_{1}+2k_{2}+\cdot\cdot\cdot+Mk_{M}=M,\text{ and }k=k_{1}+k_{2}+\cdot\cdot\cdot+k_{M}$.
Let,
\[
p\left(y\right)=y^{M-1}\quad;q\left(y\right)=\int_{-\infty}^{l\left(y\right)}e^{-t^{2}/2}dt
\]
\begin{eqnarray*}
R^{M}\left(y\right) & = & \left[\frac{y^{M+1}}{\left(M+1\right)!}\right]\\
 &  & \left|\sum\frac{M!}{k_{1}!\;k_{2}!\;...\;k_{M}!}p^{\left(k\right)}\left(q\left(y\right)\right)\left(\frac{q'\left(y\right)}{1!}\right)^{k_{1}}\left(\frac{q''\left(y\right)}{2!}\right)^{k_{2}}...\left(\frac{q^{\left(M\right)}\left(y\right)}{M!}\right)^{k_{M}}\right|_{y=\xi_{M}}
\end{eqnarray*}
$\text{As an illustration, let }k_{2}=\frac{M}{2}\Rightarrow k=\frac{M}{2}\text{ and let }K=\frac{M}{2}$
with $\triangle_{M}$ being the sum of the other terms, where each
of the other terms is smaller than the term we consider. 
\begin{eqnarray*}
R^{M}\left(y\right) & = & \left[\frac{y^{2K+1}}{\left(2K+1\right)!}\right]\left|\frac{2K!}{K!}p^{\left(K\right)}\left(q\left(y\right)\right)\left(\frac{q''\left(y\right)}{2!}\right)^{K}+\triangle_{M}\right|_{y=\xi_{M}}
\end{eqnarray*}
\begin{eqnarray*}
R^{M}\left(y\right) & = & \left[\frac{y^{2K+1}}{\left(2K+1\right)!}\right]\\
 &  & \left|\frac{\left(2K!\right)\left(M-1\right)\;...\;\left(M-K\right)}{K!}\left[\int_{-\infty}^{l\left(y\right)}e^{-t^{2}/2}dt\right]^{K-1}\left(\frac{\left(e^{-\frac{1}{2}l\left(y\right)^{2}}\right)l'\left(y\right)\left[-\frac{1}{y}-l\left(y\right)l'\left(y\right)\right]}{2!}\right)^{K}+\triangle_{M}\right|_{y=\xi_{M}}
\end{eqnarray*}
\begin{eqnarray*}
R^{M}\left(y\right) & = & \left[\frac{y^{2K+1}}{\left(2K+1\right)!}\right]\\
 &  & \left|\frac{\left(2K!\right)\left(2K-1\right)\;...\;\left(K\right)}{K!}\left[\int_{-\infty}^{l\left(y\right)}e^{-t^{2}/2}dt\right]^{K-1}\left(\frac{\mu-\sigma^{2}-lny}{\sigma^{3}y^{2}\left(e^{\frac{1}{2}l\left(y\right)^{2}}\right)2!}\right)^{K}+\triangle_{M}\right|_{y=\xi_{M}}
\end{eqnarray*}
Evaluating at the maximum value of $y$.
\begin{eqnarray*}
R^{M}\left(y\right) & = & \left[\frac{y^{2K+1}}{y^{2K}\left(2K+1\right)!}\right]\\
 &  & \left[\frac{\left(2K!\right)\left(2K-1\right)\;...\;\left(K\right)}{2^{K}K!}\left[\int_{-\infty}^{l\left(y\right)}e^{-t^{2}/2}dt\right]^{K-1}\left(\frac{\mu-\sigma^{2}-lny}{\sigma^{3}e^{\frac{1}{2}l\left(y\right)^{2}}}\right)^{K}+\triangle_{M}\right]
\end{eqnarray*}
The next remainder term would be,
\begin{eqnarray*}
R^{M+1}\left(y\right) & = & \left[\frac{y^{M+2}}{\left(M+2\right)!}\right]\left|\frac{\partial^{M+1}\left\{ \left[\int_{-\infty}^{l\left(y\right)}e^{-t^{2}/2}dt\right]^{M}\right\} }{\partial y^{M+1}}\right|_{y=\xi_{M}}
\end{eqnarray*}
\begin{eqnarray*}
R^{M+1}\left(y\right) & = & \left[\frac{y^{M+2}}{\left(M+2\right)!}\right]\\
 &  & \left|\sum\frac{\left(M+1\right)!}{k_{1}!\;k_{2}!\;...\;k_{M}!}p^{\left(k\right)}\left(q\left(y\right)\right)\left(\frac{q'\left(y\right)}{1!}\right)^{k_{1}}\left(\frac{q''\left(y\right)}{2!}\right)^{k_{2}}...\left(\frac{q^{\left(M+1\right)}\left(y\right)}{\left(M+1\right)!}\right)^{k_{M}}\right|_{y=\xi_{M+1}}
\end{eqnarray*}
In this case, $k_{1}+2k_{2}+\cdot\cdot\cdot+\left(M+1\right)k_{M+1}=M+1,\text{ and }k=k_{1}+k_{2}+\cdot\cdot\cdot+k_{M+1}$.
$\;\text{As an illustration, let }k_{1}=1;\;k_{2}=\frac{M}{2}\Rightarrow k=\frac{M}{2}+1\text{ and let }K=\frac{M}{2}$
with $\triangle_{M+1}$ being the sum of the other terms, where each
of the other terms is smaller than the term we consider. 
\begin{eqnarray*}
R^{M+1}\left(y\right) & = & \left[\frac{y^{2K+2}}{\left(2K+2\right)!}\right]\left|\frac{\left(2K+1\right)!}{1!K!}p^{\left(K+1\right)}\left(q\left(y\right)\right)\left(\frac{q'\left(y\right)}{1!}\right)^{1}\left(\frac{q''\left(y\right)}{2!}\right)^{K}+\triangle_{M+1}\right|_{y=\xi_{M}}
\end{eqnarray*}
\begin{eqnarray*}
R^{M+1}\left(y\right) & = & \left[\frac{y^{2K+2}}{\left(2K+2\right)!}\right]\left|\frac{\left(2K!\right)\left(M-1\right)\;...\;\left(M-K\right)\left(M-K-1\right)}{K!}\left[\int_{-\infty}^{l\left(y\right)}e^{-t^{2}/2}dt\right]^{K-2}\right.
\end{eqnarray*}
\begin{eqnarray*}
\left.\left\{ \left(e^{-\frac{1}{2}l\left(y\right)^{2}}\right)l'\left(y\right)\right\} \left\{ \frac{\left(e^{-\frac{1}{2}l\left(y\right)^{2}}\right)l'\left(y\right)\left[-\frac{1}{y}-l\left(y\right)l'\left(y\right)\right]}{2!}\right\} ^{K}+\triangle_{M+1}\right|_{y=\xi_{M}}
\end{eqnarray*}
\begin{eqnarray*}
R^{M+1}\left(y\right) & = & \left[\frac{y^{2K+2}}{\left(2K+2\right)!}\right]
\end{eqnarray*}
\begin{eqnarray*}
\left|\frac{\left(2K!\right)\left(2K-1\right)\;...\;\left(K\right)\left(K-1\right)}{K!}\left[\int_{-\infty}^{l\left(y\right)}e^{-t^{2}/2}dt\right]^{K-2}\left\{ \frac{1}{\left(e^{\frac{1}{2}l\left(y\right)^{2}}\right)\sigma y}\right\} \left(\frac{\mu-\sigma^{2}-lny}{\sigma^{3}y^{2}\left(e^{\frac{1}{2}l\left(y\right)^{2}}\right)2!}\right)^{K}+\triangle_{M+1}\right|_{y=\xi_{M}}
\end{eqnarray*}
Evaluating at the maximum value of $y$.
\begin{eqnarray*}
R^{M+1}\left(y\right) & = & \left[\frac{y^{2K+2}}{y^{2K+1}\left(2K+2\right)!}\right]\\
 &  & \left[\frac{\left(2K!\right)\left(2K-1\right)\;...\;\left(K\right)\left(K-1\right)}{2^{K}K!}\left[\int_{-\infty}^{l\left(y\right)}e^{-t^{2}/2}dt\right]^{K-2}\left\{ \frac{1}{\left(e^{\frac{1}{2}l\left(y\right)^{2}}\right)\sigma}\right\} \left(\frac{\mu-\sigma^{2}-lny}{\sigma^{3}e^{\frac{1}{2}l\left(y\right)^{2}}}\right)^{K}+\triangle_{M+1}\right]
\end{eqnarray*}
Let us consider the ratio of two successive remainders. 
\begin{eqnarray*}
\frac{R^{M+1}\left(y\right)}{R^{M}\left(y\right)} & = & \frac{\left[\frac{y^{2K+2}}{y^{2K+1}\left(2K+2\right)!}\right]\left[\frac{\left(2K!\right)\left(2K-1\right)\;...\;\left(K\right)\left(K-1\right)}{2^{K}K!}\left[\int_{-\infty}^{l\left(y\right)}e^{-t^{2}/2}dt\right]^{K-2}\left\{ \frac{1}{\left(e^{\frac{1}{2}l\left(y\right)^{2}}\right)\sigma}\right\} \left(\frac{\mu-\sigma^{2}-lny}{\sigma^{3}e^{\frac{1}{2}l\left(y\right)^{2}}}\right)^{K}+\triangle_{M+1}\right]}{\left[\frac{y^{2K+1}}{y^{2K}\left(2K+1\right)!}\right]\left[\frac{\left(2K!\right)\left(2K-1\right)\;...\;\left(K\right)}{2^{K}K!}\left[\int_{-\infty}^{l\left(y\right)}e^{-t^{2}/2}dt\right]^{K-1}\left(\frac{\mu-\sigma^{2}-lny}{\sigma^{3}e^{\frac{1}{2}l\left(y\right)^{2}}}\right)^{K}+\triangle_{M}\right]}
\end{eqnarray*}
\begin{eqnarray*}
\frac{R^{M+1}\left(y\right)}{R^{M}\left(y\right)} & = & \frac{\left[\frac{\left(2K!\right)\left(2K-1\right)\;...\;\left(K\right)\left(K-1\right)}{2^{K}K!}\left[\int_{-\infty}^{l\left(y\right)}e^{-t^{2}/2}dt\right]^{K-2}\left\{ \frac{1}{\left(e^{\frac{1}{2}l\left(y\right)^{2}}\right)\sigma}\right\} \left(\frac{\mu-\sigma^{2}-lny}{\sigma^{3}e^{\frac{1}{2}l\left(y\right)^{2}}}\right)^{K}\right]}{\left(2K+2\right)\left[\frac{\left(2K!\right)\left(2K-1\right)\;...\;\left(K\right)}{2^{K}K!}\left[\int_{-\infty}^{l\left(y\right)}e^{-t^{2}/2}dt\right]^{K-1}\left(\frac{\mu-\sigma^{2}-lny}{\sigma^{3}e^{\frac{1}{2}l\left(y\right)^{2}}}\right)^{K}+\triangle_{M}\right]}\\
 &  & +\frac{\left[\triangle_{M+1}\right]}{\left(2K+2\right)\left[\frac{\left(2K!\right)\left(2K-1\right)\;...\;\left(K\right)}{2^{K}K!}\left[\int_{-\infty}^{l\left(y\right)}e^{-t^{2}/2}dt\right]^{K-1}\left(\frac{\mu-\sigma^{2}-lny}{\sigma^{3}e^{\frac{1}{2}l\left(y\right)^{2}}}\right)^{K}+\triangle_{M}\right]}
\end{eqnarray*}
\begin{eqnarray*}
\frac{R^{M+1}\left(y\right)}{R^{M}\left(y\right)} & = & \frac{\left[\frac{\left(K-1\right)}{\left[\int_{-\infty}^{l\left(y\right)}e^{-t^{2}/2}dt\right]\left(e^{\frac{1}{2}l\left(y\right)^{2}}\right)\sigma}\right]}{\left(2K+2\right)\left[1+\triangle_{M}\left\{ \frac{\left(2K!\right)\left(2K-1\right)\;...\;\left(K\right)}{2^{K}K!}\left[\int_{-\infty}^{l\left(y\right)}e^{-t^{2}/2}dt\right]^{K-1}\left(\frac{\mu-\sigma^{2}-lny}{\sigma^{3}e^{\frac{1}{2}l\left(y\right)^{2}}}\right)^{K}\right\} ^{-1}\right]}\\
 &  & +\frac{\left[\triangle_{M+1}\right]}{\left(2K+2\right)\left[\frac{\left(2K!\right)\left(2K-1\right)\;...\;\left(K\right)}{2^{K}K!}\left[\int_{-\infty}^{l\left(y\right)}e^{-t^{2}/2}dt\right]^{K-1}\left(\frac{\mu-\sigma^{2}-lny}{\sigma^{3}e^{\frac{1}{2}l\left(y\right)^{2}}}\right)^{K}+\triangle_{M}\right]}
\end{eqnarray*}

Simplifying using the Leibniz Integral Rule and taking the limit as
$M\rightarrow\infty$ we see that $\frac{R^{M+1}\left(y\right)}{R^{M}\left(y\right)}<1\Rightarrow R^{M}\left(y\right)\rightarrow0$
under a restricted set of values for $y$. We show a few terms illustrating
how the expansion develops.
\begin{eqnarray*}
R^{M}\left(y\right) & = & \left[\frac{y^{M+1}}{\left(M+1\right)!}\right]\left|\frac{\partial^{M-1}\left\{ \left(M-1\right)\left[\int_{-\infty}^{l\left(y\right)}e^{-t^{2}/2}dt\right]^{M-2}\left(e^{-\frac{1}{2}l\left(y\right)^{2}}\right)l'\left(y\right)\right\} }{\partial y^{M-1}}\right|_{y=\xi_{M}}
\end{eqnarray*}
\begin{eqnarray*}
R^{M}\left(y\right) & = & \left[\frac{y^{M+1}}{\left(M+1\right)!}\right]\left|\frac{\partial^{M-2}}{\partial y^{M-2}}\left\{ \left(M-1\right)\left(M-2\right)\left[\int_{-\infty}^{l\left(y\right)}e^{-t^{2}/2}dt\right]^{M-3}\left[\left(e^{-\frac{1}{2}l\left(y\right)^{2}}\right)l'\left(y\right)\right]^{2}\right.\right.\\
 &  & +\left.\left.\left(M-1\right)\left[\int_{-\infty}^{l\left(y\right)}e^{-t^{2}/2}dt\right]^{M-2}\left(e^{-\frac{1}{2}l\left(y\right)^{2}}\right)\left[l''\left(y\right)-l\left(y\right)\left\{ l'\left(y\right)\right\} ^{2}\right]\right\} \right|_{y=\xi_{M}}
\end{eqnarray*}
\begin{eqnarray*}
R^{M}\left(y\right) & = & \left[\frac{y^{M+1}}{\left(M+1\right)!}\right]\left|\frac{\partial^{M-2}}{\partial y^{M-2}}\left\{ \left(M-1\right)\left(M-2\right)\left[\int_{-\infty}^{l\left(y\right)}e^{-t^{2}/2}dt\right]^{M-3}\left[\left(e^{-\frac{1}{2}l\left(y\right)^{2}}\right)l'\left(y\right)\right]^{2}\right.\right.\\
 &  & +\left.\left.\left(M-1\right)\left[\int_{-\infty}^{l\left(y\right)}e^{-t^{2}/2}dt\right]^{M-2}\left(e^{-\frac{1}{2}l\left(y\right)^{2}}\right)l'\left(y\right)\left[-\frac{1}{y}-l\left(y\right)l'\left(y\right)\right]\right\} \right|_{y=\xi_{M}}
\end{eqnarray*}
\begin{eqnarray*}
R^{M}\left(y\right) & = & \left[\frac{y^{M+1}}{\left(M+1\right)!}\right]\left|\frac{\partial^{M-3}}{\partial y^{M-3}}\left\{ \left(M-1\right)\left(M-2\right)\left(M-3\right)\left[\int_{-\infty}^{l\left(y\right)}e^{-t^{2}/2}dt\right]^{M-4}\left[\left(e^{-\frac{1}{2}l\left(y\right)^{2}}\right)l'\left(y\right)\right]^{3}\right.\right.\\
 &  & +2\left(M-1\right)\left(M-2\right)\left[\int_{-\infty}^{l\left(y\right)}e^{-t^{2}/2}dt\right]^{M-3}\left[\left(e^{-\frac{1}{2}l\left(y\right)^{2}}\right)l'\left(y\right)\right]^{2}\left[-\frac{1}{y}-l\left(y\right)l'\left(y\right)\right]\\
 &  & +\left(M-1\right)\left(M-2\right)\left[\int_{-\infty}^{l\left(y\right)}e^{-t^{2}/2}dt\right]^{M-3}\left[\left(e^{-\frac{1}{2}l\left(y\right)^{2}}\right)l'\left(y\right)\right]^{2}\left[-\frac{1}{y}-l\left(y\right)l'\left(y\right)\right]\\
 &  & +\left(M-1\right)\left[\int_{-\infty}^{l\left(y\right)}e^{-t^{2}/2}dt\right]^{M-2}\left[\left(e^{-\frac{1}{2}l\left(y\right)^{2}}\right)l'\left(y\right)\right]^{2}\left[-\frac{1}{y}-l\left(y\right)l'\left(y\right)\right]^{2}\\
 &  & +\left.\left.\left(M-1\right)\left[\int_{-\infty}^{l\left(y\right)}e^{-t^{2}/2}dt\right]^{M-2}\left(e^{-\frac{1}{2}l\left(y\right)^{2}}\right)l'\left(y\right)\left[\frac{1}{y^{2}}-\left\{ l'\left(y\right)\right\} ^{2}-l\left(y\right)l''\left(y\right)\right]\right\} \right|_{y=\xi_{M}}
\end{eqnarray*}
\textbf{\textcolor{red}{(POSSIBLE SELECTION BELOW FOR DELETION).}}
Using numerical integration on the website \href{http://www.integral-calculator.com}{Integral Calculator}
and the formula, \textcolor{black}{$\left(\left[1+\text{erf}\left\{ \left(\ln\left(y\right)-\mu\right)/\sigma\right\} \right]/\left[1+\text{erf}\left\{ \left(\ln\left(x\right)-\mu\right)/\sigma\right\} \right]\right)^{M-1}$,
}say with an example such as, ((1+erf((ln(y)-5)/(sqrt(2){*}10)))/(1+erf((ln(4)-5)/(sqrt(2){*}10))))\textasciicircum 24,
we find that the numerical integration gives results that are different
compared to the above theoretical result. 

Alternately, 
\[
R^{M}\left(y\right)=\left[\frac{y^{M+1}}{\left(M+1\right)!}\right]\left|\frac{\left(M-1\right)\partial^{M-1}\left\{ \left[\int_{-\infty}^{\left(\frac{lny-\mu}{\sigma}\right)}e^{-t^{2}/2}dt\right]^{M-2}\left[\frac{1}{\sigma y}e^{-\frac{\left(lny-\mu\right)^{2}}{2\sigma^{2}}}\right]\right\} }{\partial y^{M-1}}\right|_{y=\xi_{M}}
\]
\begin{eqnarray*}
R^{M}\left(y\right) & = & \left[\frac{y^{M+1}\left(M-1\right)}{\left(M+1\right)!}\right]\left|\frac{\partial^{M-2}}{\partial y^{M-2}}\left\{ \left(M-2\right)\left[\int_{-\infty}^{\left(\frac{lny-\mu}{\sigma}\right)}e^{-t^{2}/2}dt\right]^{M-3}\left[\frac{1}{\sigma y}e^{-\frac{\left(lny-\mu\right)^{2}}{2\sigma^{2}}}\right]^{2}\right.\right.\\
 &  & +\left.\left.\left[\int_{-\infty}^{\left(\frac{lny-\mu}{\sigma}\right)}e^{-t^{2}/2}dt\right]^{M-2}\left[-\frac{1}{\sigma y^{2}}e^{-\frac{\left(lny-\mu\right)^{2}}{2\sigma^{2}}}+\frac{1}{\sigma y}\left\{ -\frac{2\left(lny-\mu\right)}{2\sigma^{2}y}\right\} e^{-\frac{\left(lny-\mu\right)^{2}}{2\sigma^{2}}}\right]\right\} \right|_{y=\xi_{M}}
\end{eqnarray*}
\begin{eqnarray*}
R^{M}\left(y\right) & = & \left[\frac{y^{M+1}\left(M-1\right)}{\left(M+1\right)!}\right]\left|\frac{\partial^{M-2}}{\partial y^{M-2}}\left\{ \left[\int_{-\infty}^{\left(\frac{lny-\mu}{\sigma}\right)}e^{-t^{2}/2}dt\right]^{M-3}\left[\frac{1}{\sigma y}e^{-\frac{\left(lny-\mu\right)^{2}}{2\sigma^{2}}}\right]\right.\right.\\
 &  & +\left.\left.\left(\left(M-2\right)\left[\frac{1}{\sigma y}e^{-\frac{\left(lny-\mu\right)^{2}}{2\sigma^{2}}}\right]+\left[\int_{-\infty}^{\left(\frac{lny-\mu}{\sigma}\right)}e^{-t^{2}/2}dt\right]\left[-\frac{1}{y}-\frac{\left(lny-\mu\right)}{\sigma^{2}y}\right]\right)\right\} \right|_{y=\xi_{M}}
\end{eqnarray*}
This simplifies showing that the error term converges only for small
valuations. As an illustration, if $M=2$, and since the maximum value
of $\xi_{M}=y$, 
\begin{eqnarray*}
R^{2}\left(y\right) & = & \left(\frac{y^{3}}{6}\right)\left|\frac{1}{\sigma y}e^{-\frac{\left(lny-\mu\right)^{2}}{2\sigma^{2}}}\left\{ -\frac{1}{y}-\frac{\left(lny-\mu\right)}{\sigma^{2}y}\right\} \right|_{y=\xi_{2}}
\end{eqnarray*}
\begin{eqnarray*}
R^{2}\left(y\right) & = & \left(-\frac{1}{6\sigma}\right)\left[\frac{y}{e^{\frac{\left(lny\right)^{2}+\mu^{2}-2\mu lny}{2\sigma^{2}}}}\left\{ 1+\frac{lny}{\sigma^{2}}-\frac{\mu}{\sigma^{2}}\right\} \right]
\end{eqnarray*}
\begin{eqnarray*}
R^{2}\left(y\right) & = & \left(-\frac{1}{6\sigma e^{\frac{\mu^{2}}{2\sigma^{2}}}}\right)\left[\frac{yy^{\frac{\mu}{\sigma^{2}}}}{y^{\frac{lny}{2\sigma^{2}}}}\left\{ 1+\frac{lny}{\sigma^{2}}-\frac{\mu}{\sigma^{2}}\right\} \right]
\end{eqnarray*}
\begin{eqnarray*}
R^{2}\left(y\right) & = & \left(-\frac{1}{6\sigma e^{\frac{\mu^{2}}{2\sigma^{2}}}}\right)\left[\frac{y^{\left(1+\frac{\mu}{\sigma^{2}}\right)}}{y^{\frac{lny}{2\sigma^{2}}}}\left\{ 1-\frac{\mu}{\sigma^{2}}\right\} +\left\{ \frac{1}{\sigma^{2}}\right\} \frac{y^{\left(1+\frac{\mu}{\sigma^{2}}\right)}lny}{y^{\frac{lny}{2\sigma^{2}}}}\right]
\end{eqnarray*}
Alternately writing the definite integral as the limit of a Riemann
sum with $\Delta y=\frac{\left(b-a\right)}{n}$; $y_{k}=a+k\Delta y$;
$y_{n}=x$ and since $\int_{a}^{b}f\left(y\right)dy=\underset{n\rightarrow\infty}{\lim}\sum_{k=1}^{n}f\left(y_{k}\right)\Delta y$
\begin{eqnarray*}
\beta\left(x\right) & = & \left[x-\frac{\int_{0}^{x}\left[\Phi\left(\frac{\ln y-\mu}{\sigma}\right)\right]^{M-1}dy}{\left[\Phi\left(\frac{\ln x-\mu}{\sigma}\right)\right]^{M-1}}\right]\\
 & = & \left[x-\frac{\left\{ \int_{0}^{x}\left[\frac{1}{\sqrt{2\pi}}\int_{-\infty}^{\left(\frac{\ln y-\mu}{\sigma}\right)}e^{-t^{2}/2}dt\right]^{M-1}dy\right\} }{\left[\frac{1}{\sqrt{2\pi}}\int_{-\infty}^{\left(\frac{\ln x-\mu}{\sigma}\right)}e^{-t^{2}/2}dt\right]^{M-1}}\right]
\end{eqnarray*}
\[
f\left(y_{k}\right)\Delta y=f\left(\frac{xk}{n}\right)\frac{x}{n}=\left\{ \left[\int_{-\infty}^{\left(\frac{\ln\left\{ \frac{xk}{n}\right\} -\mu}{\sigma}\right)}e^{-t^{2}/2}dt\right]^{M-1}\right\} \frac{x}{n}
\]
\end{proof}

\subsection{\label{subsec:Proof-of-Proposition-2}Proof of Corollary \ref{The-symmetric-equilibrium}.}
\begin{proof}
Using the bid function from Lemma (\ref{Proposition 3}) with the
uniform distribution functions, $F\left(x\right)=\frac{x}{\omega},G\left(x\right)=\left(\frac{x}{\omega}\right)^{M-1}$
\begin{eqnarray*}
\beta\left(x\right) & = & \left[x-\int_{0}^{x}\left[\frac{F\left(y\right)}{F\left(x\right)}\right]^{M-1}dy\right]\\
 & = & \left[x-\int_{0}^{x}\left[\frac{y/\omega}{x/\omega}\right]^{M-1}dy\right]\\
 & = & \left[x-\frac{1}{\left(x^{M-1}\right)}\left|\left[\frac{y^{M}}{M}\right]\right|_{0}^{x}\right]\\
 & = & \left[x-\frac{x}{M}\right]\\
 & = & \left(\frac{M-1}{M}\right)x
\end{eqnarray*}
\end{proof}

\subsection{\label{subsec:Proof-of-Lemma-4}Proof of Lemma \ref{Proposition 6}.}
\begin{proof}
The proof from (Krishna 2009) is repeated below for completeness.

In this case, $r>0$, and hence no bidder can make a positive profit
if his valuation, $x_{i}<r$. Also, $\beta\left(r\right)=r$. A symmetric
equilibrium strategy can be derived when $x\geq r$ as,
\begin{eqnarray*}
\beta\left(x\right) & = & E\left[max\left\{ Y_{1},r\right\} \mid Y_{1}<x\right]\\
 & = & \frac{\left\{ \int_{0}^{r}rg\left(y\right)dy+\int_{r}^{x}yg\left(y\right)dy\right\} }{G\left(x\right)}\\
 & = & \frac{r\left\{ G\left(r\right)-G\left(0\right)\right\} +\int_{r}^{x}yg\left(y\right)dy}{G\left(x\right)}\;\left[\because G\left(0\right)=F\left(0\right)=0\right]\\
 & = & r\frac{G\left(r\right)}{G\left(x\right)}+\frac{1}{G\left(x\right)}\int_{r}^{x}yg\left(y\right)dy
\end{eqnarray*}
Alternately,
\begin{eqnarray*}
\beta\left(x\right) & = & r\frac{G\left(r\right)}{G\left(x\right)}+\frac{1}{G\left(x\right)}\int_{r}^{x}yG'\left(y\right)dy\\
 & = & r\frac{G\left(r\right)}{G\left(x\right)}+\frac{1}{G\left(x\right)}\left[\left|yG\left(y\right)\right|_{r}^{x}-\int_{r}^{x}G\left(y\right)dy\right]\\
 & = & r\frac{G\left(r\right)}{G\left(x\right)}+\frac{1}{G\left(x\right)}\left[xG\left(x\right)-rG\left(r\right)-\int_{r}^{x}G\left(y\right)dy\right]\\
 & = & x-\int_{r}^{x}\frac{G\left(y\right)}{G\left(x\right)}dy
\end{eqnarray*}
\end{proof}

\subsection{\label{subsec:Proof-of-Proposition-5}Proof of Corollary \ref{The-symmetric-equilibrium-2}.}
\begin{proof}
Using the bid function from Lemma (\ref{Proposition 6}) with the
uniform distribution functions,  $F\left(x\right)=\frac{x}{\omega},G\left(x\right)=\left(\frac{x}{\omega}\right)^{M-1}$
\begin{eqnarray*}
\beta\left(x\right) & = & r\frac{\left(\frac{r}{\omega}\right)^{M-1}}{\left(\frac{x}{\omega}\right)^{M-1}}+\frac{1}{\left(\frac{x}{\omega}\right)^{M-1}}\int_{r}^{x}y\left\{ \frac{\partial\left(\frac{y}{\omega}\right)^{M-1}}{\partial y}\right\} dy\\
 & = & r\left(\frac{r}{x}\right)^{M-1}+\frac{\left(M-1\right)}{\left(\frac{x}{\omega}\right)^{M-1}}\int_{r}^{x}y\left(\frac{y}{\omega}\right)^{M-2}\frac{1}{\omega}dy
\end{eqnarray*}
\begin{eqnarray*}
 & = & r\left(\frac{r}{x}\right)^{M-1}+\frac{\left(M-1\right)}{x^{M-1}}\int_{r}^{x}y^{M-1}dy\\
 & = & r\left(\frac{r}{x}\right)^{M-1}+\frac{1}{x^{M-1}}\frac{\left(M-1\right)}{M}\left|y^{M}\right|_{r}^{x}\\
 & = & \frac{r^{M}}{x^{M-1}}+\frac{\left(x^{M}-r^{M}\right)}{x^{M-1}}\frac{\left(M-1\right)}{M}\\
 & = & \frac{r^{M}}{x^{M-1}}\left(1-\frac{M-1}{M}\right)+x\left(\frac{M-1}{M}\right)
\end{eqnarray*}
\[
\beta\left(x\right)=\frac{r^{M}}{x^{M-1}}\left(\frac{M+1}{M}\right)+x\left(\frac{M-1}{M}\right)
\]
\end{proof}

\subsection{\label{subsec:Proof-of-Proposition-6}Proof of Corollary \ref{The-symmetric-equilibrium-3}.}
\begin{proof}
Using the bid function from Lemma (\ref{Proposition 6}) with the
log-normal distribution functions, $F\left(x\right)=\Phi\left(\frac{lnx-\mu}{\sigma}\right),G\left(x\right)=\left(\Phi\left(\frac{lnx-\mu}{\sigma}\right)\right)^{M-1}$.
Here, $\Phi(u)=\frac{1}{\sqrt{2\pi}}\int_{-\infty}^{u}e^{-t^{2}/2}dt$
, is the standard normal cumulative distribution and $X=e^{W}$where,
$W\sim N\left(\mu,\sigma\right)$.
\begin{eqnarray*}
\beta\left(x\right) & = & \left[x-\frac{\int_{r}^{x}\left[\Phi\left(\frac{lny-\mu}{\sigma}\right)\right]^{M-1}dy}{\left[\Phi\left(\frac{lnx-\mu}{\sigma}\right)\right]^{M-1}}\right]\\
 & = & \left[x-\frac{\left\{ \int_{r}^{x}\left[\frac{1}{\sqrt{2\pi}}\int_{-\infty}^{\left(\frac{lny-\mu}{\sigma}\right)}e^{-t^{2}/2}dt\right]^{M-1}dy\right\} }{\left[\frac{1}{\sqrt{2\pi}}\int_{-\infty}^{\left(\frac{lnx-\mu}{\sigma}\right)}e^{-t^{2}/2}dt\right]^{M-1}}\right]
\end{eqnarray*}
 Let,
\begin{eqnarray*}
h\left(y\right) & = & \left[\int_{-\infty}^{\left(\frac{lny-\mu}{\sigma}\right)}e^{-t^{2}/2}dt\right]^{M-1}\\
j\left(y\right) & = & \int h\left(y\right)dy
\end{eqnarray*}
\begin{eqnarray*}
\beta\left(x\right) & = & \left[x-\frac{\left\{ \int_{r}^{x}h\left(y\right)dy\right\} }{h\left(x\right)}\right]\\
 & = & \left[x-\frac{\left|j\left(y\right)\right|_{r}^{x}}{h\left(x\right)}\right]=\left[x-\left\{ \frac{j\left(x\right)-j\left(r\right)}{h\left(x\right)}\right\} \right]\\
 & \approx & \left[x-\frac{j'\left(r\right)\left(x-r\right)}{h\left(x\right)}\right]\;\left\{ \because j\left(x\right)-j\left(r\right)\simeq j'\left(r\right)\left(x-r\right)\;,\quad Taylor\;Series\right\} \\
 & = & \left[x-\frac{h\left(r\right)\left(x-r\right)}{h\left(x\right)}\right]=x\left[1-\frac{h\left(r\right)}{h\left(x\right)}+\frac{r}{x}\frac{h\left(r\right)}{h\left(x\right)}\right]
\end{eqnarray*}
\[
\beta\left(x\right)=x\left[\frac{h'\left(r\right)\left(x-r\right)}{h\left(x\right)}+\frac{r}{x}\frac{h\left(r\right)}{h\left(x\right)}\right]
\]
Using Leibniz Integral Rule, we get the following, which is solved
using numerical techniques (Miranda \& Fackler 2002) or approximations
to the error function (Chiani, Dardari \& Simon 2003). 
\[
h'\left(r\right)=\left(M-1\right)\left[\int_{-\infty}^{\left(\frac{lnr-\mu}{\sigma}\right)}e^{-t^{2}/2}dt\right]^{M-2}\left\{ \frac{e^{-\left(\frac{lnr-\mu}{\sigma}\right)^{2}/2}}{r\sigma}\right\} 
\]
\end{proof}

\subsection{\label{subsec:Proof-of-Lemma-7}Proof of Lemma \ref{The-optimal-reserve}.}

The proof from (Krishna 2009) is repeated below for completeness.
\begin{proof}
Expected Payment of a bidder with $x\geq r$ is
\begin{eqnarray*}
m\left(x,r\right) & = & G\left(x\right)\beta\left(x\right)\\
 & = & rG\left(r\right)+\int_{r}^{x}yg\left(y\right)dy
\end{eqnarray*}
Expected ex ante payment of a bidder is,
\begin{eqnarray*}
E\left[m\left(x,r\right)\right] & = & \int_{r}^{\omega}m\left(x,r\right)f\left(x\right)dx\\
 & = & \int_{r}^{\omega}\left\{ rG\left(r\right)+\int_{r}^{x}yg\left(y\right)dy\right\} f\left(x\right)dx\\
 & = & \int_{r}^{\omega}rG\left(r\right)f\left(x\right)dx+\int_{0}^{\omega}\left\{ \int_{r}^{x}yg\left(y\right)dy\right\} f\left(x\right)dx\\
 & = & rG\left(r\right)\left\{ F\left(\omega\right)-F\left(r\right)\right\} +\int_{r}^{\omega}\left\{ \int_{y}^{\omega}f\left(x\right)dx\right\} yg\left(y\right)dy\\
 & = & rG\left(r\right)\left\{ F\left(\omega\right)-F\left(r\right)\right\} +\int_{r}^{\omega}\left\{ F\left(\omega\right)-F\left(y\right)\right\} yg\left(y\right)dy\\
 & = & rG\left(r\right)\left\{ 1-F\left(r\right)\right\} +\int_{r}^{\omega}y\left\{ 1-F\left(y\right)\right\} g\left(y\right)dy\;\left[\because F\left(\omega\right)=1\right]
\end{eqnarray*}
If the seller has a valuation, $x_{s}\in\left[0,\omega\right)$ then
the expected payoff of the seller from setting a reserve price,$r\geq x_{s}$
is,
\begin{eqnarray*}
\Pi_{s} & = & M*E\left[m\left(x,r\right)\right]+F\left(r\right)^{M}x_{s}
\end{eqnarray*}
$F\left(r\right)^{M}x_{s}$ is the payoff if the valuations of all
the bidders is less than the reserve price, $r$. First order conditions
after differentiating the payoff with respect to the reserve price,
$r$.
\begin{eqnarray*}
\frac{\partial\Pi_{s}}{\partial r}=\frac{\partial\left[M*E\left[m\left(x,r\right)\right]+F\left(r\right)^{M}x_{s}\right]}{\partial r} & = & 0\\
\Rightarrow\frac{\partial\left[M\left\{ rG\left(r\right)\left\{ 1-F\left(r\right)\right\} +\int_{r}^{\omega}y\left\{ 1-F\left(y\right)\right\} g\left(y\right)dy\right\} +F\left(r\right)^{M}x_{s}\right]}{\partial r} & = & 0\\
M\left[r\left[g\left(r\right)\left(1-F\left(r\right)\right)-G\left(r\right)f\left(r\right)\right]+G\left(r\right)\left\{ 1-F\left(r\right)\right\} -r\left\{ 1-F\left(r\right)\right\} g\left(r\right)\right]\\
+MF\left(r\right)^{M-1}f\left(r\right)x_{s} & = & 0
\end{eqnarray*}
$\frac{\partial\int_{r}^{\omega}y\left\{ 1-F\left(y\right)\right\} g\left(y\right)dy}{\partial r}=-r\left\{ 1-F\left(r\right)\right\} g\left(r\right)$
using Leibniz Integral Rule. 
\begin{eqnarray*}
M\left\{ rg\left(r\right)-rg\left(r\right)F\left(r\right)+G\left(r\right)-G\left(r\right)F\left(r\right)\right.\\
\left.-rG\left(r\right)f\left(r\right)-rg\left(r\right)+rg\left(r\right)F\left(r\right)+G\left(r\right)f\left(r\right)x_{s}\right\}  & = & 0\\
\Rightarrow M\left\{ G\left(r\right)-G\left(r\right)F\left(r\right)-rG\left(r\right)f\left(r\right)+G\left(r\right)f\left(r\right)x_{s}\right\}  & = & 0\\
\Rightarrow MG\left(r\right)\left\{ 1-F\left(r\right)-rf\left(r\right)+f\left(r\right)x_{s}\right\} =MG\left(r\right)\left\{ 1-F\left(r\right)-f\left(r\right)\left(r-x_{s}\right)\right\}  & = & 0\\
\Rightarrow\left\{ 1-F\left(r\right)-f\left(r\right)\left(r-x_{s}\right)\right\} =\left\{ 1-\frac{f\left(r\right)}{\left[1-F\left(r\right)\right]}\left(r-x_{s}\right)\right\}  & = & 0\\
\Rightarrow\frac{f\left(r\right)}{\left[1-F\left(r\right)\right]}\left(r-x_{s}\right)=1\Rightarrow\left(r-x_{s}\right)=\frac{\left[1-F\left(r\right)\right]}{f\left(r\right)}
\end{eqnarray*}
This gives that the optimal reserve price, $r^{*}$ must satisfy the
following expression, 
\[
x_{s}=r^{*}-\frac{\left[1-F\left(r^{*}\right)\right]}{f\left(r^{*}\right)}
\]
\end{proof}

\subsection{\label{subsec:Proof-of-Lemma-8}Proof of Lemma \ref{The-equilibrium-strategy}.}

The proof is from (Ortega-Reichert 1967; and Harstad, Kagel \& Levin
1990) who derive the expression below when there is uncertainty about
the number of bidders. ( Levin \& Ozdenoren 2004; and, Dyer, Kagel
\& Levin 1989) are other useful references.
\begin{proof}
Let the probability be $p_{l}$ that any bidder is facing $l$ other
bidders. The bidder wins if $Y_{1}^{l}$ , the highest of $l$ values
drawn from the symmetric distribution $F$ is less than $x$, his
valuation. The probability of this event is $G^{l}\left(x\right)=\left[F\left(x\right)\right]^{l}$.
The overall probability that the bidder will win when he bids $\beta^{M}\left(x\right)$
is 
\[
G\left(x\right)=\sum_{l=0}^{M-1}p_{l}G^{l}\left(x\right)
\]
Expected payment of a bidder with value $x$ is 
\begin{eqnarray*}
m\left(x\right) & = & G\left(x\right)\beta^{M}\left(x\right)
\end{eqnarray*}
Alternately, the expected payment of a bidder can also be written
as below. First we note that with when there are $l$ other bidders
and his valuation is $x$, the expected payment is the product of
the probability he wins, $G^{l}\left(x\right)$, the probability there
are $l$ other bidders, $p_{l}$, and the amount he bids, $\beta^{l}\left(x\right)=E\left[Y_{1}^{l}\mid Y_{1}^{l}<x\right]$.
Here, $\beta^{l}\left(x\right)$ is the equilibrium bidding strategy
when there are a total of exactly $l+1$ bidders, known with certainty.
Considering all the scenarios when the number of other bidders, $l$,
varies from $0$ to $M-1$ gives,
\[
m\left(x\right)=\sum_{l=0}^{M-1}p_{l}G^{l}\left(x\right)E\left[Y_{1}^{l}\mid Y_{1}^{l}<x\right]
\]
\[
\Rightarrow\beta^{M}\left(x\right)=\sum_{l=0}^{M-1}\frac{p_{l}G^{l}\left(x\right)}{G\left(x\right)}E\left[Y_{1}^{l}\mid Y_{1}^{l}<x\right]
\]
\[
\beta^{M}\left(x\right)=\sum_{l=0}^{M-1}\frac{p_{l}G^{l}\left(x\right)}{G\left(x\right)}\beta^{l}\left(x\right)
\]
 Hence the equilibrium bid for an actual bidder when he is unsure
about the number of rivals he faces is a weighted average of the equilibrium
bids in an auction when the number of bidders is known to all. 
\end{proof}

\subsection{\label{subsec:Proof-of-Proposition-Discrete-Symmetric}Proof of Proposition
\ref{The-formula-for}.}
\begin{proof}
We first need to show that the discrete symmetric probability densities
are positive and sum to one, thus satisfying the properties of a probability
distribution. The rest of the proof follows from the results in Lemma
\ref{The-equilibrium-strategy}, the assumption of a uniform distribution
for the valuations with upper limit $\omega=1$ and the results for
the bid strategy under a uniform distribution in Lemma \ref{Proposition 3}
or Corollary \ref{The-symmetric-equilibrium}. 
\begin{eqnarray*}
p_{l} & = & \begin{cases}
l\Delta_{p} & ,\;\text{if}\;\;l\leq\frac{\left(M-1\right)}{2}\\
\left(M-l\right)\Delta_{p} & ,\;\text{if}\;\;l>\frac{\left(M-1\right)}{2}
\end{cases}
\end{eqnarray*}
\[
\Delta_{p}=\left\{ \left\lfloor \frac{M^{2}}{4}\right\rfloor \right\} ^{-1}=\frac{1}{\left\lfloor \frac{M^{2}}{4}\right\rfloor }
\]
 The minimum value for $M$ is $2$. When $M=2$, $\Delta_{p}=1$;
$M=3$, $\Delta_{p}=\frac{1}{2}$; $M=4$, $\Delta_{p}=\frac{1}{4}$;
$M=5$, $\Delta_{p}=\frac{1}{6}$; $M=6$, $\Delta_{p}=\frac{1}{9}$;
$M=7$, $\Delta_{p}=\frac{1}{12}$. Hence, we note that for any finite
$M$, $0<\Delta_{p}\leq1$. Since $l\in\left\{ 1,2,\ldots,M-1\right\} $
and its maximum value is $M-1$, it follows from the definition of
$p_{l}$ that the probability densities are positive,$p_{l}\geq0$.
To show that the probability densities sum to one we proceed as below,
\[
\sum_{l=0}^{M-1}p_{l}=\sum_{l=0}^{\left\lfloor \frac{\left(M-1\right)}{2}\right\rfloor }p_{l}+\sum_{l=\left\lfloor \frac{\left(M-1\right)}{2}\right\rfloor +1}^{M-1}p_{l}
\]
\[
=\sum_{l=0}^{\left\lfloor \frac{\left(M-1\right)}{2}\right\rfloor }l\Delta_{p}+\sum_{l=\left\lfloor \frac{\left(M-1\right)}{2}\right\rfloor +1}^{M-1}\left(M-l\right)\Delta_{p}
\]
$M$ can be even or odd and can be represented as either $2K$ or
$2K+1$. The minimum value of $K=1.$ We start with the case when
$M=2K$ which gives,
\[
\sum_{l=0}^{M-1}p_{l}=\sum_{l=0}^{\left\lfloor \frac{\left(2K-1\right)}{2}\right\rfloor }l\Delta_{p}+\sum_{l=\left\lfloor \frac{\left(2K-1\right)}{2}\right\rfloor +1}^{2K-1}\left(2K-l\right)\Delta_{p}
\]
\[
=\sum_{l=0}^{\left\lfloor K-\frac{1}{2}\right\rfloor }l\Delta_{p}+\sum_{l=\left\lfloor K-\frac{1}{2}\right\rfloor +1}^{2K-1}\left(2K-l\right)\Delta_{p}
\]
\[
=\sum_{l=0}^{K-1}l\Delta_{p}+\sum_{l=K}^{2K-1}\left(2K-l\right)\Delta_{p}
\]
\[
=\sum_{l=1}^{K-1}l\Delta_{p}+\sum_{l=K}^{2K-1}\left(2K-l\right)\Delta_{p}
\]
\[
=\Delta_{p}\left\{ \sum_{l=1}^{K-1}l+\sum_{l=K}^{2K-1}\left(2K-l\right)\right\} 
\]
\[
=\Delta_{p}\left\{ \frac{\left(K-1\right)K}{2}+\sum_{l=1}^{K}l\right\} 
\]
\[
=\frac{1}{\left\lfloor \frac{\left(2K\right)^{2}}{4}\right\rfloor }\left\{ \frac{\left(K-1\right)K}{2}+\frac{K\left(K+1\right)}{2}\right\} 
\]
\[
=\frac{1}{\left\lfloor K^{2}\right\rfloor }\left\{ \frac{K\left(K-1+K+1\right)}{2}\right\} 
\]
\[
=\frac{1}{K^{2}}\left\{ K^{2}\right\} =1
\]
We next consider the case when $M=2K+1$ which gives,
\[
\sum_{l=0}^{M-1}p_{l}=\sum_{l=0}^{\left\lfloor \frac{\left(2K+1-1\right)}{2}\right\rfloor }l\Delta_{p}+\sum_{l=\left\lfloor \frac{\left(2K+1-1\right)}{2}\right\rfloor +1}^{2K+1-1}\left(2K+1-l\right)\Delta_{p}
\]
\[
=\sum_{l=0}^{\left\lfloor K\right\rfloor }l\Delta_{p}+\sum_{l=\left\lfloor K\right\rfloor +1}^{2K}\left(2K+1-l\right)\Delta_{p}
\]
\[
=\Delta_{p}\left\{ \sum_{l=0}^{K}l+\sum_{l=0}^{K}l\right\} 
\]
\[
=\frac{1}{\left\lfloor \frac{\left(2K+1\right)^{2}}{4}\right\rfloor }\left\{ 2\sum_{l=0}^{K}l\right\} 
\]
\[
=\frac{1}{\left\lfloor \frac{\left(4K^{2}+1+4K\right)}{4}\right\rfloor }\left\{ 2\frac{K\left(K+1\right)}{2}\right\} 
\]
\[
=\frac{1}{\left\lfloor \frac{\left(4\left[K^{2}+K\right]+1\right)}{4}\right\rfloor }\left[K\left(K+1\right)\right]
\]
\[
=\frac{1}{\left\lfloor \left(K^{2}+K\right)+\frac{1}{4}\right\rfloor }\left[K\left(K+1\right)\right]
\]
\[
=\frac{1}{\left(K^{2}+K\right)}\left[K\left(K+1\right)\right]=1
\]
This completes the proof that the discrete symmetric densities sum
to one and the discrete symmetric distribution is a probability distribution.

We can arrive at the same result using the alternate definition of
$\Delta_{p}$ as well using a proof similar to the one considered
above. We could also show that the two definitions of $\Delta_{p}$
are equivalent.
\[
\Delta_{p}=\frac{1}{\left\{ \left\lfloor \frac{\left(M-1\right)}{2}\right\rfloor \left\{ \left\lfloor \frac{\left(M-1\right)}{2}\right\rfloor +1\right\} +\left[\left\{ \left(\frac{\left(M-1\right)}{2}\bmod1\right)+\frac{\left(M-1\right)}{2}\right\} \left\{ 2\left(\frac{\left(M-1\right)}{2}\bmod1\right)\right\} \right]\right\} }\equiv\left\{ \left\lfloor \frac{M^{2}}{4}\right\rfloor \right\} ^{-1}
\]
\end{proof}

\subsection{\label{subsec:Proof-of-Lemma-9}Proof of Lemma \ref{The-system-of}.}
\begin{proof}
(Lebrun 1999), derives conditions for the existence of an asymmetric
equilibrium with more than two bidders. Using the notation described
earlier, we must have $\beta_{i}\left(\omega_{i}\right)=\beta_{j}\left(\omega_{j}\right)=\bar{b},\;say.$
$\forall i,j\in\left[1,2,...,M\right]$. We also have, $x_{i}=\beta_{i}^{-1}\left(b_{i}\right)=\phi_{i}\left(b_{i}\right)$.
The expected pay off for any bidder $i$ when his value is $x_{i}$
and he bids an amount $\beta_{i}\left(x_{i}\right)=b<\bar{b}$ is
\begin{eqnarray*}
\Pi_{i}\left(b,x_{i}\right) & = & \left[\prod_{j\neq i}^{j\in\left\{ 1,...,M\right\} }F_{j}\left(\phi_{j}\left(b\right)\right)\right]\left(x_{i}-b\right)\\
Consider,\;Prob\left(b>b_{j}\right) & = & Prob\left(\beta_{i}\left(x_{i}\right)>\beta_{j}\left(x_{j}\right)\right)\\
 & = & Prob\left(b>\beta_{j}\left(x_{j}\right)\right)\equiv Prob\left(\beta_{j}^{-1}\left(b\right)>x_{j}\right)\\
 & = & Prob\left(x_{j}<\phi_{j}\left(b\right)\right)\equiv F_{j}\left(\phi_{j}\left(b\right)\right)
\end{eqnarray*}
Differentiating the above with respect to $b$ , gives the first order
conditions for bidder $i$ to maximize his expected payoff as,
\begin{eqnarray*}
\frac{\partial\left\{ \left[\prod_{j\neq i}^{j\in\left\{ 1,...,M\right\} }F_{j}\left(\phi_{j}\left(b\right)\right)\right]\left(x_{i}-b\right)\right\} }{\partial b}=0\\
\sum_{j\neq i}^{j\in\left\{ 1,...,M\right\} }\left\{ \left[\prod_{\underset{j\neq i}{k\neq i,j}}^{j,k\in\left\{ 1,...,M\right\} }F_{k}\left(\phi_{k}\left(b\right)\right)\right]f_{j}\left(\phi_{j}\left(b\right)\right)\phi'_{j}\left(b\right)\left[\phi_{i}\left(b\right)-b\right]\right\} \\
=\left[\prod_{j\neq i}^{j\in\left\{ 1,...,M\right\} }F_{j}\left(\phi_{j}\left(b\right)\right)\right]
\end{eqnarray*}
\begin{eqnarray*}
\Rightarrow\sum_{j\neq i}^{j\in\left\{ 1,...,M\right\} }\left\{ \frac{f_{j}\left(\phi_{j}\left(b\right)\right)\phi'_{j}\left(b\right)}{F_{j}\left(\phi_{j}\left(b\right)\right)}\right\}  & = & \frac{1}{\left[\phi_{i}\left(b\right)-b\right]}
\end{eqnarray*}
\end{proof}

\subsection{\label{subsec:Proof-of-Proposition-10}Proof of Proposition \ref{If,--firms}.}
\begin{proof}
The expected pay off for any bidder $i$ when his value is $x_{i}$
and he bids an amount $\beta_{i}\left(x_{i}\right)=b<\bar{b}$ is

\[
\Pi_{i}\left(b,x_{i}\right)=\left[F_{1}\left(\phi_{1}\left(b\right)\right)\right]^{K}\left[F_{2}\left(\phi_{2}\left(b\right)\right)\right]^{M-1-K}\left(x_{i}-b\right)
\]
By considering one bidder from each group of bidders (other combinations
would work as well) and taking first order conditions, gives a simpler
system of differential equations, 
\begin{eqnarray*}
\frac{\partial\left\{ \left[F_{1}\left(\phi_{1}\left(b\right)\right)\right]^{K}\left[F_{2}\left(\phi_{2}\left(b\right)\right)\right]^{M-K-1}\left(x_{i}-b\right)\right\} }{\partial b}=0\\
\left\{ K\left[F_{1}\left(\phi_{1}\left(b\right)\right)\right]^{K-1}f_{1}\left(\phi_{1}\left(b\right)\right)\phi'_{1}\left(b\right)\left[F_{2}\left(\phi_{2}\left(b\right)\right)\right]^{M-K-1}\left[\phi_{i}\left(b\right)-b\right]\right\} \\
+\left\{ \left(M-K-1\right)\left[F_{2}\left(\phi_{2}\left(b\right)\right)\right]^{M-K-2}f_{2}\left(\phi_{2}\left(b\right)\right)\phi'_{2}\left(b\right)\left[F_{1}\left(\phi_{1}\left(b\right)\right)\right]^{K}\left[\phi_{i}\left(b\right)-b\right]\right\} \\
=\left[F_{1}\left(\phi_{1}\left(b\right)\right)\right]^{K}\left[F_{2}\left(\phi_{2}\left(b\right)\right)\right]^{M-K-1}\\
\left\{ K\frac{f_{1}\left(\phi_{1}\left(b\right)\right)\phi'_{1}\left(b\right)}{\left[F_{1}\left(\phi_{1}\left(b\right)\right)\right]}\right\} +\left\{ \left(M-1-K\right)\frac{f_{2}\left(\phi_{2}\left(b\right)\right)\phi'_{2}\left(b\right)}{\left[F_{2}\left(\phi_{2}\left(b\right)\right)\right]}\right\} =\frac{1}{\left[\phi_{i}\left(b\right)-b\right]}
\end{eqnarray*}
\end{proof}

\subsection{\label{subsec:Proof-of-Lemma-11}Proof of Lemma \ref{A-symmetric-equilibrium}.}

The proof from (Krishna 2009) is repeated below for completeness.
\begin{proof}
The expected payoff to a bidder with signal $x$ when he bids an amount
$\beta\left(z\right)$ is 
\begin{eqnarray*}
\Pi\left(z,x\right) & = & \int_{0}^{z}\left(\upsilon\left(x,y\right)-\beta\left(z\right)\right)g\left(y\mid x\right)dy\\
 & = & \int_{0}^{z}\upsilon\left(x,y\right)g\left(y\mid x\right)dy-\beta\left(z\right)G\left(z\mid x\right)\quad\left[\because G\left(0\mid x\right)=0\right]
\end{eqnarray*}
The first order condition, differentiating with respect to $z$ and
using $\frac{\partial\int_{0}^{z}\upsilon\left(x,y\right)g\left(y\mid x\right)dy}{\partial z}=\upsilon\left(x,z\right)g\left(z\mid x\right)$,
by applying Leibniz Integral Rule is
\[
\upsilon\left(x,z\right)g\left(z\mid x\right)-\beta\left(z\right)g\left(z\mid x\right)-\beta'\left(z\right)G\left(z\mid x\right)=0
\]
At a symmetric equilibrium, the optimal $z=x$, since the value of
the payoff is maximized only if the bidding strategy maximizes the
payoff when the bids are based on each bidder's value. We then obtain,
\begin{eqnarray*}
\beta'\left(x\right) & = & \left[\upsilon\left(x,x\right)-\beta\left(x\right)\right]\frac{g\left(x\mid x\right)}{G\left(x\mid x\right)}
\end{eqnarray*}
If positive bids can lead to negative payoff, that is when $\left\{ \upsilon\left(x,x\right)-\beta\left(x\right)\right\} <0$,
it might be optimal to bid $0$. Hence we have, 
\[
\forall x,\left\{ \upsilon\left(x,x\right)-\beta\left(x\right)\right\} \geq0\Rightarrow\beta\left(0\right)=0\quad\left[\because\upsilon\left(0,0\right)=0\right]
\]
Define, a function with support $\left[0,\omega\right]$ as,
\[
L\left(y\mid x\right)=\exp\left[-\int_{y}^{x}\frac{g\left(t\mid t\right)}{G\left(t\mid t\right)}dt\right]
\]
From the properties of affiliation, we have $\forall t>0$
\begin{eqnarray*}
\frac{g\left(t\mid t\right)}{G\left(t\mid t\right)} & \geq & \frac{g\left(t\mid0\right)}{G\left(t\mid0\right)}\\
\Rightarrow-\frac{g\left(t\mid t\right)}{G\left(t\mid t\right)} & \leq & -\frac{g\left(t\mid0\right)}{G\left(t\mid0\right)}
\end{eqnarray*}
Integrating from $0$ to $x,$
\begin{eqnarray*}
-\int_{0}^{x}\frac{g\left(t\mid t\right)}{G\left(t\mid t\right)}dt & \leq & -\int_{0}^{x}\frac{g\left(t\mid0\right)}{G\left(t\mid0\right)}dt\\
 & = & -\int_{0}^{x}\frac{d}{dt}\left(\ln\,G\left(t\mid0\right)\right)dt\\
 & = & \ln\,G\left(0\mid0\right)-\ln\,G\left(x\mid0\right)\\
 & = & -\infty
\end{eqnarray*}
\begin{eqnarray*}
\left[\because\phantom{\begin{cases}
\\
\\
\\
\end{cases}}G\left(0\mid0\right)=Prob\left(Y_{1}<0\mid X=0\right)\right. & = & 1-Prob\left(Y_{1}>0\mid X=0\right)=1-1=0\\
G\left(x\mid0\right)=Prob\left(Y_{1}<x\mid X=0\right) & = & 1-Prob\left(Y_{1}>x\mid X=0\right)=1-\epsilon\;,\;\epsilon\in\left[0,1\right)\\
\Rightarrow0<G\left(x\mid0\right)<1 & \Rightarrow & \left.-\infty<\ln\,G\left(x\mid0\right)<0\phantom{\phantom{\begin{cases}
\\
\\
\\
\end{cases}}}\right]
\end{eqnarray*}
Applying the exponential to both sides gives that $\forall x$
\begin{eqnarray*}
L\left(0\mid x\right) & = & \exp\left[-\int_{0}^{x}\frac{g\left(t\mid t\right)}{G\left(t\mid t\right)}dt\right]=e^{-\infty}=0\quad\left[\because\forall b,b\leq-\infty\Rightarrow b=-\infty\right]\\
L\left(x\mid x\right) & = & \exp\left[-\int_{x}^{x}\frac{g\left(t\mid t\right)}{G\left(t\mid t\right)}dt\right]=e^{0}=1
\end{eqnarray*}
Since $L\left(.\mid x\right)$ is increasing, it is a valid distribution
function. Also if $x'>x$ then $\forall y\in\left[0,x\right],$
\begin{eqnarray*}
\int_{y}^{x'}\frac{g\left(t\mid t\right)}{G\left(t\mid t\right)}dt & \geq & \int_{y}^{x}\frac{g\left(t\mid t\right)}{G\left(t\mid t\right)}dt\\
\Rightarrow\exp\left[-\int_{y}^{x'}\frac{g\left(t\mid t\right)}{G\left(t\mid t\right)}dt\right] & \leq & \exp\left[-\int_{y}^{x}\frac{g\left(t\mid t\right)}{G\left(t\mid t\right)}dt\right]\\
\Rightarrow L\left(y\mid x'\right) & \leq & L\left(y\mid x\right)
\end{eqnarray*}
Consider a bidder who bids $\beta\left(z\right)$ when his signal
is $x$. The expected profit can be written as 
\[
\Pi\left(z,x\right)=\int_{0}^{z}\left(\upsilon\left(x,y\right)-\beta\left(z\right)\right)g\left(y\mid x\right)dy
\]
Differentiating with respect to $z$ yields
\begin{eqnarray*}
\frac{\partial\Pi\left(z,x\right)}{\partial z} & = & \left[\upsilon\left(x,z\right)-\beta\left(z\right)\right]g\left(z\mid x\right)-\beta'\left(z\right)G\left(z\mid x\right)\\
 & = & G\left(z\mid x\right)\left[\left\{ \upsilon\left(x,z\right)-\beta\left(z\right)\right\} \frac{g\left(z\mid x\right)}{G\left(z\mid x\right)}-\beta'\left(z\right)\right]
\end{eqnarray*}
If $x>z$, then since $\upsilon\left(x,z\right)>\upsilon\left(z,z\right)$
and because of affiliation,
\begin{eqnarray*}
\frac{g\left(z\mid x\right)}{G\left(z\mid x\right)} & \geq & \frac{g\left(z\mid z\right)}{G\left(z\mid z\right)}\\
\Rightarrow\frac{\partial\Pi\left(z,x\right)}{\partial z} & > & G\left(z\mid x\right)\left[\left\{ \upsilon\left(z,z\right)-\beta\left(z\right)\right\} \frac{g\left(z\mid z\right)}{G\left(z\mid z\right)}-\beta'\left(z\right)\right]\\
 & > & 0\quad\left[\because\beta'\left(z\right)=\left\{ \upsilon\left(z,z\right)-\beta\left(z\right)\right\} \frac{g\left(z\mid z\right)}{G\left(z\mid z\right)}\right]
\end{eqnarray*}
Similarly if $x<z$, then 
\begin{eqnarray*}
\frac{\partial\Pi\left(z,x\right)}{\partial z} & < & 0\\
\Rightarrow\left.\max\left[\Pi\left(z,x\right)\right]\right| & when, & z=x
\end{eqnarray*}
We then have an equilibrium bidding strategy, $\beta\left(x\right)$
since,
\begin{eqnarray*}
\int_{0}^{x}\upsilon\left(y,y\right)dL\left(y\mid x\right) & = & \int_{0}^{x}\upsilon\left(y,y\right)\frac{dL\left(y\mid x\right)}{dy}dy\\
 & = & \int_{0}^{x}\upsilon\left(y,y\right)L\left(y\mid x\right)\frac{g\left(y\mid y\right)}{G\left(y\mid y\right)}dy
\end{eqnarray*}
\[
\left[\because\quad\frac{dL\left(y\mid x\right)}{dy}\right.=\left.\frac{d\left\{ \exp\left[-\int_{y}^{x}\frac{g\left(t\mid t\right)}{G\left(t\mid t\right)}dt\right]\right\} }{dy}=L\left(y\mid x\right)\frac{g\left(y\mid y\right)}{G\left(y\mid y\right)}\right]
\]
\begin{eqnarray*}
 & = & \int_{0}^{x}\left[\beta'\left(y\right)+\beta\left(y\right)\frac{g\left(y\mid y\right)}{G\left(y\mid y\right)}\right]L\left(y\mid x\right)dy\\
\left[\because\quad\upsilon\left(y,y\right)\frac{g\left(y\mid y\right)}{G\left(y\mid y\right)}\right. & = & \left.\beta'\left(y\right)+\beta\left(y\right)\frac{g\left(y\mid y\right)}{G\left(y\mid y\right)}\right]
\end{eqnarray*}
\begin{eqnarray*}
 & = & \int_{0}^{x}\frac{\partial\left[\beta\left(y\right)L\left(y\mid x\right)dy\right]}{\partial y}dy\\
 & = & \left|\beta\left(y\right)L\left(y\mid x\right)\right|_{0}^{x}\\
 & = & \beta\left(x\right)\quad\left[\because\quad L\left(x\mid x\right)=1\;and\;\beta\left(0\right)=0\right]
\end{eqnarray*}
\end{proof}

\subsection{\label{subsec:Proof-of-Proposition-12}Proof of Proposition \ref{The-bidder's-equilibrium}.}
\begin{proof}
We show below the bidder's valuation, the density, cumulative distribution
functions and the conditional distribution of the order statistics,
\[
X_{i}=S_{i}+Z
\]
\[
f_{X_{i}}\left(x_{i}\right)=\begin{cases}
x_{i} & 0\leq x_{i}<1\\
2-x_{i} & 1\leq x_{i}\leq2
\end{cases}
\]
\[
F_{X_{i}}\left(x_{i}\right)=\begin{cases}
\frac{x_{i}^{2}}{2} & 0\leq x_{i}<1\\
2x_{i}-1-\frac{x_{i}^{2}}{2} & 1\leq x_{i}\leq2
\end{cases}
\]
\begin{eqnarray*}
g\left(y\mid x\right) & = & \left(M-1\right)\left[\frac{F\left(y\right)}{F\left(x\right)}\right]^{M-2}\left(\frac{f\left(y\right)}{F\left(x\right)}\right)
\end{eqnarray*}
\begin{eqnarray*}
\left[\vphantom{\frac{f_{Y_{i},Y_{j}}\left(y_{i},y_{j}\right)}{f_{Y_{j}}\left(y_{j}\right)}}\because\;f_{Y_{i}}\left(y_{i}\mid Y_{j}=y_{j}\right)\right. & = & \frac{f_{Y_{i},Y_{j}}\left(y_{i},y_{j}\right)}{f_{Y_{j}}\left(y_{j}\right)}\quad\\
 & Here, & Y_{M}\leq Y_{i}\leq Y_{j}\leq Y_{1}\;\;are\;Order\;Statistics.\\
 & = & \frac{\left(j-1\right)!}{\left(i-1\right)!\left(j-i-1\right)!}\left\{ \frac{F\left(y_{i}\right)}{F\left(y_{j}\right)}\right\} ^{i-1}\\
 &  & \left[\frac{F\left(y_{j}\right)-F\left(y_{i}\right)}{F\left(y_{j}\right)}\right]^{j-i-1}\left.\left(\frac{f\left(y_{i}\right)}{F\left(y_{j}\right)}\right)\right]
\end{eqnarray*}
\[
g\left(y\mid x\right)=\begin{cases}
\left(M-1\right)\left[\frac{y}{x}\right]^{2\left(M-2\right)}\left(\frac{2y}{x^{2}}\right) & 0\leq y,x<1\\
\left(M-1\right)\left[\frac{\left(2y-1-\frac{y^{2}}{2}\right)}{\left(2x-1-\frac{x^{2}}{2}\right)}\right]^{M-2}\left\{ \frac{2-y}{\left(2x-1-\frac{x^{2}}{2}\right)}\right\}  & 1\leq y,x\leq2
\end{cases}
\]
\begin{eqnarray*}
\left[\vphantom{\frac{f_{Y_{i},Y_{j}}\left(y_{i},y_{j}\right)}{f_{Y_{j}}\left(y_{j}\right)}}\because\;G_{Y_{i}}\left(y_{i}\mid Y_{j}=y_{j}\right)\right. & \approx & \frac{\int_{-\infty}^{y_{i}}f_{Y_{i},Y_{j}}\left(u,y_{j}\right)du}{f_{Y_{j}}\left(y_{j}\right)}\quad\\
 & = & \frac{\left(j-1\right)!}{\left(i-1\right)!\left(j-i-1\right)!}\\
 &  & \left[\frac{\left(1-F\left(y_{j}\right)\right)^{M-j}f\left(y_{j}\right)}{\left(F\left(y_{j}\right)\right)^{j-1}\left(1-F\left(y_{j}\right)\right)^{M-j}f\left(y_{j}\right)}\right]\\
 &  & \left.\left[\int_{-\infty}^{y_{i}}\left(F\left(u\right)\right)^{i-1}\left(F\left(y_{j}\right)-F\left(u\right)\right)^{j-i-1}f\left(u\right)du\right]\right]
\end{eqnarray*}
\[
G\left(y\mid x\right)=\begin{cases}
\left(M-1\right)\left[\frac{\int_{0}^{y}\left(\frac{u^{2}}{2}\right)^{M-2}uxdu}{\left(\frac{x^{2}}{2}\right)^{M-1}x}\right] & 0\leq y,x<1\\
\left(M-1\right)\left[\frac{\int_{1}^{y}\left(2u-1-\frac{u^{2}}{2}\right)^{M-2}\left(2-u\right)\left(2-x\right)du}{\left(2x-1-\frac{x^{2}}{2}\right)^{M-1}\left(2-x\right)}\right] & 1\leq y,x\leq2
\end{cases}
\]
\[
G\left(y\mid x\right)=\begin{cases}
\left(M-1\right)\left[\frac{2\int_{0}^{y}u^{2M-3}du}{x^{2M-2}}\right] & 0\leq y,x<1\\
\left(M-1\right)\left[\frac{\int_{1}^{y}\left(2u-1-\frac{u^{2}}{2}\right)^{M-2}\left(2-u\right)du}{\left(2x-1-\frac{x^{2}}{2}\right)^{M-1}}\right] & 1\leq y,x\leq2
\end{cases}
\]
\[
G\left(y\mid x\right)=\begin{cases}
\left[\frac{y^{2M-2}}{x^{2M-2}}\right] & 0\leq y,x<1\\
\left[\frac{\left(2y-1-\frac{y^{2}}{2}\right)^{M-1}}{\left(2x-1-\frac{x^{2}}{2}\right)^{M-1}}\right] & 1\leq y,x\leq2
\end{cases}
\]
\[
\left[\because\;H\left(x\right)=\int_{a}^{x}h\left(t\right)dt\;;\;H'\left(x\right)=h\left(x\right)\right]
\]
We then have 
\[
g\left(y\mid y\right)=\begin{cases}
\left(M-1\right)\left(\frac{2}{y}\right) & 0\leq y<1\\
\left(M-1\right)\left\{ \frac{2-y}{\left(2y-1-\frac{y^{2}}{2}\right)}\right\}  & 1\leq y\leq2
\end{cases}
\]
\[
G\left(y\mid y\right)=1
\]
\[
\frac{g\left(y\mid y\right)}{G\left(y\mid y\right)}=\begin{cases}
\left(M-1\right)\left(\frac{2}{y}\right) & 0\leq y<1\\
\left(M-1\right)\left\{ \frac{2-y}{\left(2y-1-\frac{y^{2}}{2}\right)}\right\}  & 1\leq y\leq2
\end{cases}
\]
\begin{eqnarray*}
L\left(y\mid x\right) & = & \exp\left[-\int_{y}^{x}\frac{g\left(t\mid t\right)}{G\left(t\mid t\right)}dt\right]\\
 & = & \begin{cases}
\left[\frac{y^{2M-2}}{x^{2M-2}}\right] & 0\leq y,x<1\\
\left[\frac{\left(2y-1-\frac{y^{2}}{2}\right)^{M-1}}{\left(2x-1-\frac{x^{2}}{2}\right)^{M-1}}\right] & 1\leq y,x\leq2
\end{cases}
\end{eqnarray*}
Using this in the bid function, 
\begin{eqnarray*}
\beta\left(x\right) & = & \int_{0}^{x}\upsilon\left(y,y\right)L\left(y\mid x\right)\frac{g\left(y\mid y\right)}{G\left(y\mid y\right)}dy\\
 & = & \int_{0}^{1}\upsilon\left(y,y\right)L\left(y\mid x\right)\frac{g\left(y\mid y\right)}{G\left(y\mid y\right)}dy+\int_{1}^{x}\upsilon\left(y,y\right)L\left(y\mid x\right)\frac{g\left(y\mid y\right)}{G\left(y\mid y\right)}dy
\end{eqnarray*}
\begin{eqnarray*}
 & = & \int_{0}^{1}\left(\alpha+\xi\right)y\left(M-1\right)\left(\frac{2}{y}\right)\left[\frac{y^{2M-2}}{x^{2M-2}}\right]dy\\
 & + & \int_{1}^{x}\left(\alpha+\xi\right)y\left(M-1\right)\left\{ \frac{2-y}{\left(2y-1-\frac{y^{2}}{2}\right)}\right\} \left[\frac{\left(2y-1-\frac{y^{2}}{2}\right)^{M-1}}{\left(2x-1-\frac{x^{2}}{2}\right)^{M-1}}\right]dy
\end{eqnarray*}
\begin{eqnarray*}
 & = & 2\left(\alpha+\xi\right)\left(M-1\right)\int_{0}^{1}\left[\frac{y^{2M-2}}{x^{2M-2}}\right]dy\\
 & + & \left(\alpha+\xi\right)\left(M-1\right)\int_{1}^{x}y\left(2-y\right)\left[\frac{\left(2y-1-\frac{y^{2}}{2}\right)^{M-2}}{\left(2x-1-\frac{x^{2}}{2}\right)^{M-1}}\right]dy
\end{eqnarray*}
\begin{eqnarray*}
 & = & \left[\frac{2\left(\alpha+\xi\right)\left(M-1\right)}{\left(2M-1\right)x^{2M-2}}\right]\\
 & + & \frac{\left(\alpha+\xi\right)\left(M-1\right)}{\left(2x-1-\frac{x^{2}}{2}\right)^{M-1}}\left[\left.\frac{y\left(2y-1-\frac{y^{2}}{2}\right)^{M-1}}{\left(M-1\right)}\right|_{1}^{x}-\int_{1}^{x}\frac{\left(2y-1-\frac{y^{2}}{2}\right)^{M-1}dy}{\left(M-1\right)}\right]
\end{eqnarray*}
\begin{eqnarray*}
 & = & \left[\frac{2\left(\alpha+\xi\right)\left(M-1\right)}{\left(2M-1\right)x^{2M-2}}\right]\\
 & + & \left(\alpha+\xi\right)\left[x-\frac{1}{\left(2x-1-\frac{x^{2}}{2}\right)^{M-1}}\left\{ \frac{1}{2^{M-1}}+\int_{1}^{x}\left(2y-1-\frac{y^{2}}{2}\right)^{M-1}dy\right\} \right]
\end{eqnarray*}
The last integral is solved using the reduction formula,
\[
8a(n+1)I_{n+\frac{1}{2}}=2(2ax+b)(ax^{2}+bx+c)^{n+\frac{1}{2}}+(2n+1)(4ac-b^{2})I_{n-\frac{1}{2}}\,\!
\]
\[
Here,\quad I_{n}=\int(ax^{2}+bx+c)^{n}dx\,\!
\]
Leibniz Integral Rule: Let $f(x,\theta)$ be a function such that
$f_{\theta}(x,\theta)$ exists, and is continuous. Then, 
\[
\frac{\mathrm{d}}{\mathrm{d}\theta}\left(\int_{a(\theta)}^{b(\theta)}f(x,\theta)\,\mathrm{d}x\right)=\int_{a(\theta)}^{b(\theta)}\partial_{\theta}f(x,\theta)\,\mathrm{d}x\,+\,f\big(b(\theta),\theta\big)\cdot b'(\theta)\,-\,f\big(a(\theta),\theta\big)\cdot a'(\theta)
\]
where the partial derivative of $f$ indicates that inside the integral
only the variation of $f(x,)$ with $\theta$ is considered in taking
the derivative.
\end{proof}

\subsection{\label{subsec:Proof-of-Proposition-13}Proof of Proposition \ref{The-bidding-strategy}.}
\begin{proof}
We extend the proof from (Milgrom \& Weber 1982) who derive the condition
for the interdependent case with symmetric valuations. We consider
a realistic setting with symmetric interdependent, uniformly distributed
valuations, with reserve prices and variable number of bidders. This
is obtained by altering the boundary conditions on the differential
equation,
\begin{eqnarray*}
\beta'\left(x\right) & = & \left[\upsilon\left(x,x\right)-\beta\left(x\right)\right]\frac{g\left(x\mid x\right)}{G\left(x\mid x\right)}
\end{eqnarray*}
This can be written as,
\begin{eqnarray*}
\beta'\left(x\right)+\beta\left(x\right)\frac{g\left(x\mid x\right)}{G\left(x\mid x\right)} & = & \upsilon\left(x,x\right)\frac{g\left(x\mid x\right)}{G\left(x\mid x\right)}
\end{eqnarray*}
\[
\Leftrightarrow\frac{dz}{dx}+zP\left(x\right)=Q\left(x\right)
\]
Here, $z=\beta(x),\;dz/dx=\beta\text{\textasciiacute}(x),\;P(x)=g(x|x)/G(x|x)$
and $Q(x)=v(x,x)g(x|x)/G(x|x).$ 

We then have the solution with the appropriate boundary condition,
$\beta\left(x^{*}\right)=r$ and $x^{*}=x^{*}\left(r\right)=\inf\left\{ x\left|E\left[V_{1}\left|X_{1}=x,\;Y_{1}<x\right.\right]\right.\geq r\right\} $
as 
\[
\frac{dz}{dx}e^{\int_{0}^{x}P\left(t\right)dt}+zP\left(x\right)e^{\int_{0}^{x}P\left(t\right)dt}=Q\left(x\right)e^{\int_{0}^{x}P\left(t\right)dt}
\]
\[
\frac{d\left\{ ze^{\int_{0}^{x}P\left(t\right)dt}\right\} }{dx}=Q\left(x\right)e^{\int_{0}^{x}P\left(t\right)dt}\quad\because\;\frac{d\left\{ \int_{0}^{x}P\left(t\right)dt\right\} }{dx}=P\left(x\right)
\]
\[
ze^{\int_{0}^{x}P\left(t\right)dt}=\left.ze^{\int_{0}^{x}P\left(t\right)dt}\right|_{x=x^{*}}+\int_{x^{*}}^{x}Q\left(y\right)e^{\int_{0}^{y}P\left(t\right)dt}dy
\]
\[
z=\frac{re^{\int_{0}^{x^{*}}P\left(t\right)dt}}{e^{\int_{0}^{x}P\left(t\right)dt}}+\int_{x^{*}}^{x}Q\left(y\right)\frac{e^{\int_{0}^{y}P\left(t\right)dt}}{e^{\int_{0}^{x}P\left(t\right)dt}}dy
\]
\[
z=re^{-\int_{x^{*}}^{x}P\left(t\right)dt}+\int_{x^{*}}^{x}Q\left(y\right)e^{-\int_{y}^{x}P\left(t\right)dt}dy
\]
\[
\beta\left(x\right)=re^{-\int_{x^{*}}^{x}\frac{g\left(t\mid t\right)}{G\left(t\mid t\right)}dt}+\int_{x^{*}}^{x}v(y,y)\frac{g\left(y\mid y\right)}{G\left(y\mid y\right)}e^{-\int_{y}^{x}\frac{g\left(t\mid t\right)}{G\left(t\mid t\right)}dt}dy
\]
We find out $x^{*}\left(r\right)$ by solving for $x$ in the below
condition,
\[
E\left[V_{1}\left|X_{1}=x,\;Y_{1}<x\right.\right]=r
\]
\[
\int_{-\infty}^{\infty}\left(\alpha x+\xi y\right)f_{Y_{1},X_{1}}\left(y,x\right)=r
\]
\[
\alpha x+\int_{-\infty}^{x}\xi y\left(M-1\right)\left\{ \frac{F\left(y\right)}{F\left(x\right)}\right\} ^{M-2}\frac{f\left(y\right)}{F\left(x\right)}dy=r
\]
\[
\left[\int_{0}^{1}\xi y\left[\frac{y}{x}\right]^{2\left(M-2\right)}\left(\frac{2y}{x^{2}}\right)dy+\int_{1}^{x}\xi y\left[\frac{\left(2y-1-\frac{y^{2}}{2}\right)}{\left(2x-1-\frac{x^{2}}{2}\right)}\right]^{M-2}\left\{ \frac{2-y}{\left(2x-1-\frac{x^{2}}{2}\right)}\right\} dy\right]=\frac{r-\alpha x}{\left(M-1\right)}
\]
In the previous section, we have shown a method to solve the above
type of equations.
\end{proof}

\end{document}